\documentclass[aps,prb,superscr iptaddress,amsfonts,amsmath,amssymb,floatfix,reprint,longbibliography]{revtex4-2}

\usepackage[utf8]{inputenc}
\usepackage{amsmath}
\usepackage{amsthm}
\usepackage{mathtools}
\usepackage{amsfonts}
\usepackage{todonotes}
\usepackage{comment}
\usepackage{stmaryrd}
\usepackage[ruled]{algorithm2e}
\usepackage{algcompatible}
\usepackage{centernot}
\usepackage{times}
\usepackage{color}
\usepackage{float}

\usepackage{natbib}
\setcitestyle{numbers,round,comma}

\usepackage{mdframed}
\usepackage{bbm}
\usepackage{tikz}

\newtheorem{theorem}{Theorem}[section]
\newtheorem{corollary}[theorem]{Corollary}

\newtheorem{definition}[theorem]{Definition}

\newcommand{\Prob}{\operatorname{Pr}}

\newcommand{\indexf}{\mathbbm{1}}
\newcommand{\bin}{\operatorname{binary}}

\usepackage{xcolor}
\definecolor{mygreen}{HTML}{00FFD1}
\definecolor{myblue}{HTML}{31C6D4}
\definecolor{myyellow}{HTML}{FFFF00}
\definecolor{myred}{HTML}{FF1E1E}
\usetikzlibrary{shapes}
\usetikzlibrary{automata, positioning, arrows}

\renewcommand{\mod}{\text{ mod }}

\DeclarePairedDelimiter\ceil{\lceil}{\rceil}

\newcommand{\lsb}{\operatorname{LSB}}

\newcommand{\ntext}[1]{#1}
\usepackage[normalem]{ulem}
\newcommand{\otext}[1]{}

\newcommand{\bsn}{\{0,1\}^n}

\newcommand{\Eval}{Eval}
\newcommand{\Con}{Con}
\newcommand{\opt}{opt_{\Con}}

\newcommand{\error}{\operatorname{error}}

\newcommand{\sizeS}{\left|S\right|}
\newcommand{\sizeh}{\left| h \right|}

\newcommand{\BC}{\operatorname{BC}}

\newcommand{\RSA}{\operatorname{RSA}}
\newcommand{\CRSA}{\operatorname{C-RSA}}

\newcommand{\powers}{\operatorname{powers}}
\newcommand{\bcinput}{\bin\left(\powers_N(\RSA(x,N,e)),N,e\right)}
\newcommand{\DFA}{\operatorname{DFA}}
\newcommand{\DFARSA}{\operatorname{DFA-RSA}}
\newcommand{\BF}{\operatorname{BF}}
\newcommand{\BFRSA}{\operatorname{BF-RSA}}
\newcommand{\LSTM}{\operatorname{LSTM}}
\newcommand{\LSTMRSA}{\operatorname{LSTM-RSA}}
\newcommand{\FC}{\operatorname{FC}}
\newcommand{\FCRSA}{\operatorname{FC-RSA}}
\newcommand{\ILP}{\operatorname{ILP}}
\newcommand{\ILPRSA}{\operatorname{ILP-RSA}}

\date{5.12.2022}

\SetKwInOut{Input}{Input}\SetKwInOut{Output}{Output}

\begin{document}

\title{An in-principle super-polynomial quantum advantage for approximating\\ combinatorial optimization problems \ntext{via computational learning theory}}

\author{Niklas Pirnay}
\address{Electrical Engineering and Computer Science, Technische Universit{\"a}t Berlin,  10587 Berlin, 
Germany}
\author{Vincent Ulitzsch}
\address{Electrical Engineering and Computer Science, Technische Universit{\"a}t Berlin,  10587 Berlin, 
Germany}
\author{Frederik Wilde}
\affiliation{Dahlem Center for Complex Quantum Systems, Freie Universit\"{a}t Berlin, 14195 Berlin, Germany}
\author{Jens Eisert}
\affiliation{Dahlem Center for Complex Quantum Systems, Freie Universit\"{a}t Berlin, 14195 Berlin, Germany}
\affiliation{Fraunhofer Heinrich Hertz Institute, 10587 Berlin, Germany}
\author{Jean-Pierre Seifert}
\address{Electrical Engineering and Computer Science, Technische Universit{\"a}t Berlin, 10587 Berlin,
Germany}
\affiliation{Fraunhofer SIT, Rheinstra{\ss}e 75, 64295 Darmstadt, Germany}

\maketitle
{
\bf It is unclear to what extent quantum algorithms can outperform classical algorithms for problems of combinatorial optimization. In this work, by resorting to computational learning theory and cryptographic notions, we \ntext{give a fully constructive} proof that quantum computers feature a super-polynomial advantage over classical computers in approximating combinatorial optimization problems. Specifically, by building on seminal work by Kearns and Valiant, we \otext{identify} \ntext{provide} special instances that are hard for classical computers to approximate up to polynomial factors. Simultaneously, we give a quantum algorithm that can efficiently approximate the optimal solution within a polynomial factor. The quantum advantage in this work is ultimately borrowed from Shor's quantum algorithm for factoring. We introduce an explicit and comprehensive end-to-end construction for the advantage bearing instances. For such instances, quantum computers have, in principle, the power to approximate combinatorial optimization solutions beyond the reach of classical efficient algorithms.
}

\section{Introduction}

Recent years have enjoyed an enormous interest in quantum computing as a new paradigm of computing.
While ground breaking work \cite{shor_factoring_1994, montanaro_algorithms_2016, arute_googlesupremacy_2019} established that quantum computers provide a \otext{significant} \ntext{substantial} speedup for certain problems over classical computers, the extent of this \textit{quantum advantage} is still largely uncharted territory.
It has been suggested that quantum computers may actually assist in improving existing classical algorithms  for the task of \emph{combinatorial optimization}.
\ntext{That is, the task of assigning discrete values from a finite set to finitely-many variables, such that the cost function over the variables is minimal.}
\otext{To this date, no proof for a quantum advantage for this task has been given.}
Here, we provide a full constructive proof that quantum computers can indeed outperform classical computers for finding approximations to combinatorial optimization problems.

Combinatorial optimization problems arise in a wealth of contexts,
ranging from problems in the description of nature to industrial resource optimization
\cite{cook_optimization_1997}. In combinatorial optimization problems, one is given an objective function which needs to be optimized over a finite set of object, such that some constraints over the objects are satisfied.
A prominent example is the \emph{travelling salesperson problem}, 
in which one has to choose a cyclic route through a set of cities, such that the length of the route is minimal (see Fig.~1\textbf{A}).
In this case, the objective function is the sum of travelled distances along the route, which needs to be minimized. The objects are the cities and the constraints demand that the start- and endpoints are the same city and no city is visited twice. 
The travelling salesperson problem underlies many routing problems that we encounter in our every-day life, such as finding the most efficient supply chain, the cheapest delivery route or the fastest 3D print.
But also job scheduling, resource allocation, or portfolio optimization---and many naturally occurring problems such as that of protein folding---can basically be seen as combinatorial optimization problems.
Given the vast social and economic significance of combinatorial optimization problems, it is not surprising that they have been a subject of intense research for many decades.
However, many problems of this kind are known to be NP-hard in worst case complexity, i.e., even the best algorithms to date cannot solve all instances of combinatorial optimization problems in tractable time.
This does not mean that one cannot solve practically relevant instances up to reasonable system sizes or find good approximations to the optimal solutions.
There is indeed a rich body of literature on both heuristic approaches that work well in practice \cite{hromkovic_heuristics_2004}
as well as on a rigorous theory of approximating solutions \cite{williamson_design_of_approximation_algorithms_2011}.
For example, enormous traveling salesperson instances of up to 85.900 cities have been solved optimally \cite{applegate_tsp_2009} and there are many software suites that enable good approximations for the industry today.

Motivated by the insight that \emph{quantum computers} may offer substantial computational speedups over classical computers \cite{shor_factoring_1994, montanaro_algorithms_2016},
it has long been suggested that quantum computers 
may actually assist in further improving approximations to such problems.
While there is no hope for an efficient either quantum or classical 
algorithm that is \emph{guaranteed to find the optimal solution}, \otext{the} \ntext{a crucially important} question is whether quantum computers offer an advantage for combinatorial optimization problems and specifically for \textit{approximating} the solution of such problems.

This topic is particularly prominently discussed in the realm of \emph{near-term quantum computers}  \cite{preskill_quantum_2018}, for which full
quantum error correction and fault tolerance seem 
out of scope, but which may well offer
computational advantages over classical computers
\cite{arute_googlesupremacy_2019,hangleiter_supremacyreview_2022}. 
Indeed, for such devices, algorithms such as the \emph{quantum approximate optimization algorithm} \cite{farhi_qaoa_2014}
have been designed precisely to solve combinatorial optimization problems of the above mentioned kind.
Surely these instances of variational algorithms 
\cite{cerezo_variational_2021, mcclean_variational_2016,zhou_qaoaperformance_2020} 
will not always be able to solve such problems: 
At best, these algorithms may be able to produce approximate solutions
that are better than those found by classical computers.
They may also be able to efficiently find good approximations for more instances
than classical computers when they are used perfectly.
When actually operated in realistic, noisy environments, the performance of quantum devices is further reduced. 
Indeed, for variational algorithms run on noisy devices, some obstacles have 
been identified for quantum computers that involve circuits that are deeper 
than logarithmic \cite{stilck_noisyoptimization_2020, ryuji_errormitigation_2022, quek_errormitigation_2022, takagi_errormitigation_2022},
obstacles that may well be read as indications that it will be
challenging to achieve quantum advantages 
in the presence of realistic noise levels.

\begin{figure*}
    \centering

\includegraphics[width=.9\textwidth]{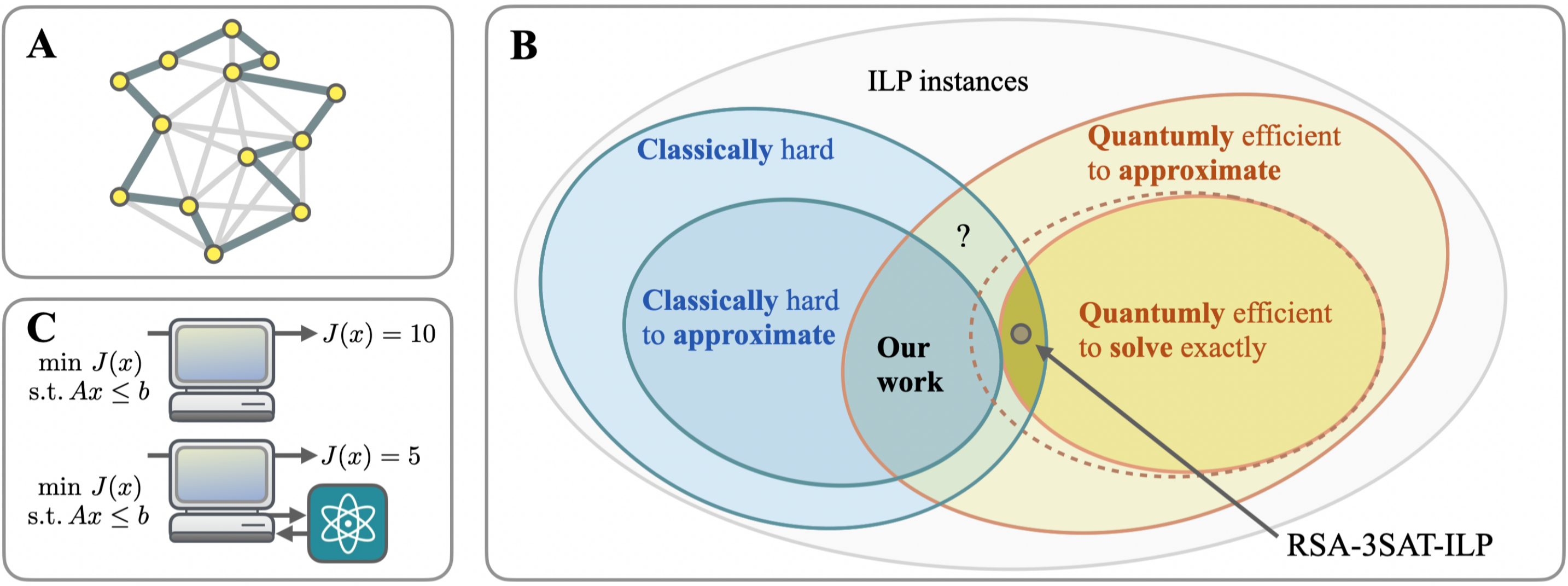}\\
    \caption{ \ntext{\textbf{Overview of the setting of our work.}} \textbf{A} A diagrammatic sketch of the travelling salesperson
    problem aimed at finding the
    shortest possible route that visits each city (represented
    as vertices)
    exactly once and returns to the origin city.
    \textbf{B} Venn diagram that depicts the sense in which a quantum advantage---symbolized in \textbf{C}---is proven in our work for integer programming problems. The grey set contains all instances of integer linear programming, and the subsets contain the hard or respectively easy to solve instances. By hard to approximate we mean that there is no polynomial time algorithm that approximates the size of the optimal solution up to a factor of $opt^{\alpha} \cdot |I|^{\beta}$, where $|I|$ is the instance size, $\alpha, \beta$ are constants such that $\alpha \geq 0$, $0 \leq \beta < 1$ and $opt$ denotes the size of the optimal solution.
    Whether the dotted line holds true, i.e., whether there exists a problem that can be solved \textit{exactly} by a polynomial-time quantum algorithm, but are hard to approximate classically, is left for further research.
    }
    \label{fig:Venn}
\end{figure*}

For variational algorithms aimed at tackling classical 
combinatorial optimization problems
that are being cast in the 
form of minimizing the
energy of commuting Hamiltonian terms, 
further obstructions are known
\cite{garcia_errorpropagation_2022}. Some small instances of the
problem can even be classically efficiently simulated (even
though small noise levels may help
\cite{liu_smallnoise_2022}).

The make-or-break question, therefore, is: 
What is, after all, the potential of quantum computers for tackling combinatorial optimization problems?
A simple quantum advantage for exactly solving combinatorial optimization problems may be obtained by reducing the integer factoring problem to 3-SAT and leveraging the advantage of Shor's algorithm \cite{shor_factoring_1994}.
\otext{A quantum advantage for \textit{approximating} the solution of combinatorial optimization problems seems much more difficult to obtain.}
\ntext{A quantum advantage for \textit{approximating} the solution of combinatorial optimization problems can also be obtained using a different proof technique than used in this manuscript.
As outlined in Ref.\ \cite{szegedy_optimization_2022}, the celebrated PCP theorem can be used to show classical approximation hardness, while Shor's algorithm for factoring can be used for an efficient quantum approximation algorithm.
Thus, an in-principle separation between classical and quantum approximation algorithms can already be obtained from the PCP machinery and Shor's algorithm.
However, the focus of this work is to provide a technically detailed and complete proof, that comprehensively describes the reductions and gives and end-to-end guidelines on how to construct the advantage-bearing combinatorial optimization instances.
We expect that, the concrete realization of the proof gives follow-up work additional insights over a generic proof sketch.}
\otext{Given the practical importance of such a task and its wide applicability, this poses a highly relevant question.}
\ntext{Given the practical importance of combinatorial optimization tasks and its wide applicability, this is a valuable contribution to further advance quantum optimization algorithms.}

\section{Results}

\subsection{Premise of this work}

\otext{In this work, we provide a comprehensive answer to the affirmative: A fault tolerant quantum computer can approximate certain combinatorial optimization problems super-polynomially more efficiently than a classical computer.}
\ntext{In this work, we provide a comprehensive proof that a fault tolerant quantum computer can approximate certain combinatorial optimization problems super-polynomially more efficiently than a classical computer.  While such a result can 
also be obtained from the PCP theorem and Shor's algorithm \cite{szegedy_optimization_2022}, our 
work focuses on fully fleshing out a constructive proof, in order to provide a clear guideline on how such advantage-bearing instances can be constructed.
An important contribution of our work -- in particular in the light of claims of applications of quantum computers for solving optimization problems that have become common -- is also in contributing to clarifying in what precise sense one can hope for quantum advantages in optimization in the first place.}

In our efforts, we \ntext{digress from the PCP theorem and} build on the work of Ref.\ \cite{kearns_boolfunc_1993}, who have shown the classical hardness of approximating the solution of the so-called \emph{formula colouring problem}, a combinatorial optimization problem which generalizes the \emph{graph colouring problem}.
We continue to draw inspiration from Ref.~\cite{kearns_boolfunc_1993} when showing an approximation hardness preserving reduction from the formula colouring problem to \emph{integer programming} (a family of combinatorial optimization problems
on which variants of quantum approximation have already been applied to \cite{deller_quiditqaoa_2022}).
To prove the super-polynomial quantum advantage, we extend the work of Ref.~\cite{kearns_boolfunc_1993} to show the classical approximation hardness for certain integer programming instances that are constructed from the \emph{RSA encryption function}. We then provide an efficient quantum algorithm for approximating the solutions of those instances up to a polynomial factor. \ntext{For a given instance $\mathcal{I}$ of integer programming or formula colouring, it can be decided in quantum polynomial time whether $\mathcal{I}$ belongs to this set of advantage bearing instances.}

We also formulate the hard-to-approximate instances in the optimization problem of minimizing the energy of commuting Hamiltonian terms, connecting our findings to the widely studied field of variational quantum optimization.
Since the classical approximation hardness stems from the hardness of inverting the 
\emph{RSA encryption function} \cite{kearns_boolfunc_1993}, the core of the quantum advantage discovered in this work is ultimately essentially borrowed from, once again, 
Shor's quantum algorithm \cite{shor_factoring_1994} for factoring.

The kind of reasoning developed here resembles the mindset of Refs.~\cite{sweke_generator_2021,pirnay_density_2022,liu_supervised_2021}
to the problem of approximating solutions to combinatorial optimization.
The argument we have put forth 
compellingly shows that quantum computers  
can indeed perform provably substantially better than classical computers on instances of approximating
combinatorial optimization problems, in fact, featuring a super-polynomial speedup. 
To make contact with quantum approximate optimization, we also spell out how the
problem instances can be written in terms \ntext{of} Hamiltonian optimization.
While the results found here are highly motivating and do show the potential of
quantum devices to tackle such practically relevant problems, it remains open to which extent
this potential can be unlocked for short variational quantum circuits as they
are accessible in near-term quantum computers.

This  result is interesting due to the technical aspects in its own right---showcasing the potential of quantum computers to offer speedups when tackling combinatorial optimization problems. It is also interesting conceptually, because it provides guidance on the question what type of speedups one can expect from further quantum approximation algorithms.
The present work does \emph{not} suggest to solve NP-hard problems
exactly on a quantum computer in polynomial time.
Instead, \otext{we show} \ntext{we provide a full proof for} an in-principle quantum advantage for classically hard-to-approximate combinatorial optimization problems \ntext{and along the way introduce a polynomial reduction strategy.}
This can be seen as a positive result on the potential
use of fault tolerant quantum computers and, possibly,
variational quantum algorithms to address such problems.

\subsection{Technical results}
Technically, in this work, we show a quantum-classical separation for the computational task of approximating combinatorial optimization problems.
To show this, one needs a set of combinatorial optimization problem instances that are classically hard-to-approximate but for which we provide an efficient quantum approximation algorithm. For the classically hard-to-approximate problem instances, we build on the work of Ref.~\cite{kearns_boolfunc_1993}, who have shown the classical hardness of approximating the solution of the so-called \emph{formula colouring problem}, a combinatorial optimization problem which generalizes the \emph{graph colouring problem}.
Before we proceed with the quantum efficiency part, we want to briefly explain the formula colouring problem and how classical approximation hardness for specific instances can be obtained.

The formula colouring problem is defined over a formula $F$ with the integer variables $z_1, \dots, z_m \in \mathbb{N}$. The value of a variable acts as the \textit{colour} of the variable.
A \textit{$k$-colouring} is an assignment of colours to the $z_i$, described by a partitioning $P$ of the variable set into $k$ equivalence classes, such that two variables are in the same partition if and only if they are assigned the same colour.
We write $z_i = z_j$ if and only if the two variables are assigned the same colour, and hence they are in the same partition in $P$.
Otherwise, we write $z_i \neq z_j$.
Let us now give a formal definition of the formula colouring problem.
\begin{definition}[Formula colouring problem $\FC$ \citep{kearns_boolfunc_1993}]
   \label{def:fc}
    \textbf{Instance} A Boolean formula $F(z_1, \dots, z_m)$ which consists of conjunctions of clauses of the form either $(z_i \neq z_j)$ or the form $((z_i = z_j) \rightarrow (z_k = z_l))$.\\
    \textbf{Solution} A minimal colouring $P$ for $F(z_1, \dots, z_m)$ such that $F$ is satisfied.
\end{definition}
A \textit{minimal colouring} to the FC problem is a colouring with the fewest colours, i.e., $|P|$ is minimal for all possible colourings such that $F$ is satisfied.
To internalize, consider the example formula
\begin{equation}
    (z_1 \neq z_2) \land ((z_1 = z_3) \rightarrow  (z_2 = z_4))
\end{equation}
which has the 4-colouring $\{\{z_1\}, \{z_2\}, \{z_3\}, \{z_4\}\}$ satisfying the formula and has the minimal colouring $\{\{z_1, z_3\}, \{z_2, z_4\}\}$ using only two colours while satisfying the formula.
It can be easily seen, that one can encode the graph colouring problem into the formula colouring problem by constructing a formula that only consists of clauses $(z_i \neq z_j)$ for each edge in the graph between nodes $z_i$ and $z_j$.
Thus the formula colouring problem belongs to the computationally hard-to-solve class of NP-complete problems.
In this work, we show a quantum advantage for a specific subset of formula colouring problems, that are provably hard to even approximate, but for which we present an efficient quantum approximation algorithm.
Further we give an approximation-preserving reduction from the formula colouring problem to the \textit{integer linear programming} (ILP) problem, thus showing also a quantum advantage for integer programming.
\begin{definition}[Integer linear programming problem ($\ILP$)]
    \textbf{Instance} A \ntext{linear} objective function $J$ over integer variables subject to \ntext{linear} constraints of the variables.\\
    \textbf{Solution} A valid assignment $\mathcal{A}$ of the variables under the constraints, such that the objective function $J(\mathcal{A})$ is minimal for all assignments that satisfy the constraints.
\end{definition}

So what is this subset of classically hard-to-approximate FC/ILP problem instances?
Ref.~\cite{kearns_boolfunc_1993} show how one can cleverly encode the \emph{deterministic finite automaton} (DFA) that decrypts an RSA-ciphertext into the formula colouring problem. That is to say, they show how to construct a set formula colouring problem instances $\FCRSA$, where if one would be able to find the smallest (or even approximately small) colouring, then one would be able to learn a DFA that could decrypt RSA ciphertexts.
Since decrypting RSA ciphertexts is assumed to be intractable for classical computers, when the secret key is unknown, it follows that approximating the solutions to $\FCRSA$ must be intractable.
\ntext{In this work, we substantially extend this result to ILP problems by defining a subset $\ILPRSA$ by means of a polynomial, approximation-preserving reduction of $\FCRSA$ to ILP.}
\otext{The same also applies for the problem instance set $\ILPRSA$, which is derived from $\FCRSA$, by reducing the formula colouring problem instances to integer linear programming.}
For the detailed description on how the hard-to-approximate instances are constructed and an in-depth explanation of why they are hard-to-approximate, we refer the reader to the methods sections.
\ntext{Specifically, Section~\ref{sec:quantum_efficiency} presents an overview of the chain of reductions and further hardness results derived in Ref.~\cite{kearns_boolfunc_1993}.}
%
\otext{Before we state the main results in this work, let us give the formal definition of integer linear programming problems.}

Building on the machinery developed in  Ref.~\cite{kearns_boolfunc_1993}, we prove the classical approximation hardness for the combinatorial optimization task of integer programming, i.e., for the specific subset of problem instances called $\ILPRSA$. 
As described before the instances in $\ILPRSA$ cleverly encode the decryption of an RSA ciphertext, for which the secret cryptographic key is unknown.
Hence, we obtain the following theorem, which must hold if inverting the RSA encrpytion function is computationally intractable for classical algorithms.

\begin{theorem}[Classical hardness of approximation for \textit{integer linear programming}]
   \label{theo:classical_hardness_ilprsa_maintext}
    Assuming the hardness of inverting the RSA function,
    there exists no classical probabilistic polynomial-time algorithm
    that on input an instance $\ILP_F$ of $\ILPRSA$ finds an assignment $\mathcal{A}$ of the variables in $\ILP_F$ which satisfies all constraints and approximates the optimal objective value $opt_{\ILP}(\ILP_F)$ by
    \begin{equation}
        J(\mathcal{A}) \leq opt_{\ILP} (\ILP_F)^{\alpha} |\ILP_F|^{\beta}
    \end{equation}
    for any $\alpha \geq 1$ and $0 \leq \beta < 1/4$.
\end{theorem}
The quantity $opt_{\ILP}(\ILP_F)$ is the minimal objective function value possible under the constraints in $\ILP_F$ and $|\ILP_F|$ denotes the size of the problem instance in some fixed encoding.
The Theorem above essentially states that there is no classical algorithm that finds an assignment $\mathcal{A}$ such that the objective value $J(\mathcal{A})$ is upper bounded by some polynomial in $opt_{\ILP}(\ILP_F)$ times a pre-factor that is determined by the size of the problem.
That is under the assumption that inverting the RSA function is not possible in polynomial time on a classical computer.

However, we show that there does exist a polynomial-time quantum algorithm that finds an assignment $\mathcal{A}$ of the variables that satisfies the constraints in $\ILP_F$ such that the objective value is smaller than some polynomial in $opt_{\ILP}(\ILP_F)$.
\begin{theorem}[Quantum efficiency for $\ILPRSA$]
    There exists a polynomial-time quantum algorithm that, on input an instance $\ILP_{F}$ of $\ILPRSA$, finds a variable assignment
    $\mathcal{A}$ that satisfies all constraints and for which the objective function is bounded as
    $$J(\mathcal{A}) \leq opt_{\ILP}(\ILP_{F_S})^{\alpha}$$
    for all $\ILP_{F}$ and for some $\alpha \geq 1$.
\end{theorem}
Essentially, the efficient quantum algorithm cleverly reads out the RSA parameters from an instance $\ILP_F$ of $\ILPRSA$ and then runs Shor's algorithm for integer factorization, thereby reconstructing the secret RSA key. Given the RSA secret key, the algorithm can find an assignment of the variables in $\ILP_F$ such that the objective function is a polynomial in $opt_{\ILP}(\ILP_{F_S})$.
This yields the sought after super-polynomial quantum advantage for approximating the optimal solution of combinatorial optimization problems.
The nature of this advantage is illustrated in Figs.~\ref{fig:Venn}(b) and (c).
\ntext{Note, that the factor $\vert\ILP_F\vert^\beta$ in the hardness result (Theorem~\ref{theo:classical_hardness_ilprsa_maintext}) cannot decrease the approximation gap, since $\vert\ILP_F\vert^\beta \geq 1$ for all $\beta \in [0, 1/4)$.}

The quantum algorithm presented is distinctly not of a variational type, as they are 
commonly proposed for approximating combinatorial optimization tasks using a quantum computer \cite{farhi_qaoa_2014}.
It is still meaningful to formulate the optimization problem as an energy minimization 
problem, to closely connect our findings to the performance of variational quantum algorithms \cite{cerezo_variational_2021,mcclean_variational_2016} in near-term quantum computing.
In the methods section we give the construction on how the ILP at hand can be stated in terms of a \emph{quadratic unconstrained binary optimization problems}.
All such problems can be directly mapped to Hamiltonian problems where 
the optimal objective value 
is equivalent with the ground state energy of the 
\emph{quantum Ising Hamiltonian}.

\section{Discussion}

In this work, we have made substantial progress on the important question of what potential quantum computers may
offer for approximating the solution of combinatorial optimization problems. Given the social and economic impact of such problems and the large body of the recent literature on near-term
quantum computing focusing on use cases of this kind, this is an important question.

We actually address this question from a fresh and unorthodox perspective. Equipped with tools from  mathematical cryptography, 
\ntext{and materializing the Occam's Razor framework in the
reduction -- hence settling an open question 
-- we technically present here,}
we prove a super-polynomial speedup for approximating the solution of instances
of NP-hard combinatorial optimization problems using a fault tolerant quantum computer.
We explicitly show such speedups for instances of the much discussed integer linear programming 
which are proven to be hard to approximate by classical computations.

\otext{
Here, we answer the question of the potential of quantum computers to achieve substantial 
speed-ups for combinatorial optimization problems to the affirmative. In particular, we show that quantum computers can efficiently approximate combinatorial optimization problems that cannot be efficiently approximated by classical computations up to the same accuracy. This can be seen as an 
encouraging result for a family of quantum algorithms that could derive new applications from Shor's algorithm. It goes without saying that our findings cannot be used to make strong
claims of a similar potential for variational quantum computations: Here, hard 
methods development is still required to come to such conclusions.
}

\ntext{In this work, we provide the end-to-end construction of the advantage bearing instances, 
allowing further work to gain valuable insights into the quantum advantage for combinatorial optimization.
Such instances are expected to prove to be useful to compare quantum versus classical optimization algorithms and provide a fruitful arena for future research in this field.}
\otext{It will also be interesting to fully flesh out follow-up proof sketches that come to similar conclusions based on the \emph{PCP theorem} \cite{szegedy_optimization_2022} and compare the obtained results.}
The work here shows \ntext{and provides guidance for the discussion of} what one can reasonably hope for when discussing the potential of near-term quantum algorithms to tackle problems of combinatorial optimization.

{\small 
\onecolumngrid
\vspace{\columnsep}

\section{Materials and methods}

\subsection{Preliminaries}

\subsubsection{Notation \ntext{and acronyms}}
For what follows, some notation will be required. We will heavily build on literature from the cryptographic context, and hence make use of substantial notation that is common in this context.
By $\{0,1\}^n$ we will denote the
set of $n$-bit strings, whereas $\{0,1\}^*$ are arbitrary finite length bit strings.
$2^X$ is the power set of $X$, for $X$ being a set. $\indexf(a)$ is the 
\emph{indicator function} which equates to $1$ if $a$ is true and $0$ otherwise.
$LSB(x)$ is the least significant bit of $x$.
$\mathbb{Z}_N$ is the residue class ring $\mathbb{Z}/N\mathbb{Z}$.
The application of the function $\bin(x_1,\dots,x_k)$ explicitly converts 
its inputs $x_1,\dots,x_k$ to a single coherent bit string using some fixed binary encoding.
\ntext{The result established in this work is based on a series of reductions between various classes of computational problems and brings them together in a fresh fashion.
While each of the terms is introduced explicitly in the subsequent sections, we summarize them here in a table for the reader's convenience.}

\begin{table}[H]
    \centering
    \begin{tabular}{r|l l}
         \textsc{acronym} & \textsc{meaning} & \textsc{section/definition} \\
         \hline
         \textit{Eval} & Evaluation problem & Def.~\ref{def:eval} \\
         $\Con$ & Consistency problem & Def.~\ref{def:conprob} \\
         $\opt(S)$ & Size of the minimal consistent representation class & \ref{subsec:learning_representations} \\
         A.P.R & approximation-preserving reduction & \ref{subsec:reductions} \\
         $\RSA$ & Rivest–Shamir–Adleman asymmetric cryptosystem & \ref{subsec:rsa} \\
         $\lsb$ & Least significant bit & \ref{subsec:rsa} \\
         $\BC$ & Poly-size log-depth Boolean circuits & \ref{subsec:reductions} \\
         $\CRSA$ & Boolean circuits inverting RSA & Def.~\ref{definition:rsa_learning_problem} \\
         $\DFA$ & Class of deterministic finite automata & \ref{subsubsec:dfa} \\
         $\DFARSA$ & Subclass of $\DFA$s computing the $\lsb$ of $\RSA$ & \ref{subsubsec:classical-approximation-hardness} \\
         $\FC$ & Formula coloring problems & Def.~\ref{def:fc} \\
         $\FCRSA$ & Subclass of $\FC$ encoding the solution to $\Con(\DFARSA, \DFA)$ & \ref{subseq:classhard_fc} \\
         $\BF$ & Boolean formulas & \ref{subsec:reductions} \\
         $\BFRSA$ & Subclass of $\BF$ computing the $\lsb$ of $\RSA$ & \ref{subsubsec:classical-approximation-hardness} \\
         $\LSTM$ & Log-space Turing machine & \ref{subseq:classhard_fc} \\
         $\LSTMRSA$ & subclass of $\LSTM$ computing the $\lsb$ of $\RSA$ & \ref{subsubsec:classical-approximation-hardness} \\
         $\ILP$ & Integer linear programs & Def.~\ref{def:ilp} \\
         $\ILPRSA$ & Subclass of $\ILP$ encoding the solution to $\FCRSA$ & \ref{subsubsec:classical-approximation-hardness} \\
    \end{tabular}
    \caption{An overview of notation and acronyms used in this manuscript.}
    \label{tab:acronyms}
\end{table}

\subsubsection{\ntext{Deterministic finite automata}}
\label{subsubsec:dfa}
\emph{Deterministic finite automata} (DFA) \cite{hopcroft_automata_2013} are models in computation theory, utilized for modeling systems with a finite number of states.
A DFA is formally defined as a quintuple $(Q, \Sigma, \lambda, q_0, \omega)$, where
\begin{itemize}
    \item $Q$ is a finite set of states.
    \item $\Sigma$ is a finite set of symbols, constituting the automaton's alphabet.
    \item $\lambda: Q \times \Sigma \rightarrow Q$ is the transition function.
    \item $q_0 \in Q$ represents the start state.
    \item $\omega \subseteq Q$ denotes the set of accept states.
\end{itemize}

The DFA operates on a string composed of symbols from $\Sigma$. Beginning from the start state $q_0$, it transitions between states according to the transition function $\lambda$. Upon processing the entire string, if the DFA is in a state that is part of $\omega$, the string accepted by the DFA; otherwise, it is rejected.
Figure \ref{fig:dfa} presents a graphical illustration of an exemplary DFA.
\begin{figure}[h]
\centering
\begin{tikzpicture}[shorten >=1pt,node distance=2cm,on grid,auto] 
   \node[state,initial] (q_0)   {$q_0$}; 
   \node[state] (q_1) [above right=of q_0] {$q_1$}; 
   \node[state,accepting] (q_2) [below right=of q_0] {$q_2$};
    \path[->] 
    (q_0) edge  node {a} (q_1)
          edge  node [swap] {b} (q_2)
    (q_1) edge  node  {a} (q_2)
          edge [loop above] node {b} ()
    (q_2) edge  node [swap] {a,b} (q_0);
\end{tikzpicture}
\caption{\textbf{Example of a deterministic finite automaton.} The DFA is represented as a quintuple $(Q, \Sigma, \lambda, q_0, \omega)$, where $Q = \{q_0, q_1, q_2\}$, $\Sigma = \{a, b\}$, $\lambda$ is defined by the transitions (e.g., $\lambda(q_0, a) = q_1$, $\lambda(q_0, b) = q_2$, etc.), $q_0$ is the initial state, and $\omega = \{q_2\}$ is the set of accept states.}
\label{fig:dfa}
\end{figure}
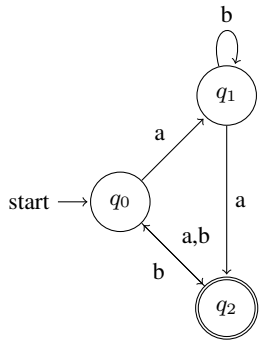
It is known that DFAs recognize exactly the set of regular languages~\cite{hopcroft_automata_2013}.

\subsubsection{Representation classes}
To show a quantum-classical separation for a computational task, one needs a classically hard problem.
Many computational tasks that are hard for classical computers may be derived from cryptography, where it can be shown that under \textit{cryptographic assumptions} (such as ``factoring is hard") learning certain concepts or properties about a specific cryptographic function is hard.
In particular, in this work we are concerned with how these concepts are represented and how large these representations are.
To do this, let us introduce the notion of representation classes 
that capture the model of concepts in a precise manner.
Let $X \subseteq \{0,1\}^*$ be a set of binary strings with finite length, called the \textit{domain} which encodes all objects of interest to us.
For example, $X$ may be the set of all images or $X$ may be the set of all music songs.
A \textit{concept} over $X$ is described by a subset of $X$ which is defined via $\{x\in X \mid \text{concept is true for } x \}$.
A concept may be for example ``depicts a tree" or ``is a happy song".

We are in particular interested in how a concept is represented. Different representations for a concept can, for example, be \emph{Boolean circuits}, \emph{Boolean formulae}, \emph{Turing machines} 
or \emph{deterministic finite automata} (DFA) \cite{hopcroft_automata_2013}.
We, therefore, define a \textit{representation class over $X$} to be the pair $(\sigma, C)$, where $C \subseteq \{0,1\}^*$ is the set of representation descriptions, for example the set of descriptions for Boolean circuits or finite automata. The function 
$\sigma : C \rightarrow 2^X$ maps a representation description to a concept. For example, 
$\sigma$ maps a DFA to the set of bit strings that it accepts or a 
Boolean formula to its satisfying assignments. We will sometimes denote $(\sigma, C)$ simply by $C$ if $\sigma$ is clear from the context.
Fig.~\ref{fig:representations} visualizes the relationship between representations and concepts.

Observe that for all $c \in C$, $\sigma(c)$ is a concept over X and the entire image space $\sigma(C)$ is called the \textit{concept class} represented by the representation class $(\sigma, C)$. We denote by $|c|$ the length of the representation description using some standard encoding. Additionally, for a representation $c \in C$, we denote by $c(x) = \indexf(\text{if } x \in \sigma(c))$ the \textit{label} of x under the concept $\sigma(c)$, with the index function $\indexf$.
Furthermore, a \textit{labeled sample} 
\begin{equation}
S=\{\left(x, c(x)\right) \mid x \in \tilde{X} \subseteq X \}
\end{equation}
of a concept $\sigma(c) $
is a set of \textit{labeled examples} from a subset $\tilde{X}$ of the domain $X$. Note that a sample consists of multiple examples.
Finally, let $(\phi, H)$ be another representation class over $X$ and let $D$ be a probability distribution over $X$. For any $h \in H$, we define the error of $h$ under $D$ with respect to a target representation $c$ as 
\begin{equation}
\error_{c,D}(h) = \Prob_{x \sim D}[c(x) \neq h(x)].
\end{equation}

\subsubsection{Polynomial-time reductions}
\label{subsec:reductions}
Polynomial-time reductions are an important building block of this work, as they will be an integral part of our proof of a quantum-classical computational separation for combinatorial optimization problems.
Building on the work of
Ref.~\cite{kearns_learningtheory_1994}, reductions are required to ``carry over" classical hardness results of representation learning to combinatorial optimization problems.
At the same time, we find that quantum computers can break the 
construction and lead to a quantum advantage for combinatorial optimization.
Let us now introduce the notion of general polynomial-time reductions among computational problems.

Let $A,B$ be two computational problems.
Consider the function $\tau$ to map an instance $\mathcal{A}$ of $A$ to an instance $\tau(\mathcal{A})$ of $B$.
Furthermore, let $g$ be a function that maps from the solution space of $B$ to the solution space of $A$.
The pair of functions $(\tau,g)$ is a \textit{polynomial-time reduction} from $A$ to $B$, if $\tau,g$ are computable in polynomial time and if $y_{\mathcal{B}}$ is a solution of $\tau(\mathcal{A})$ if and only if $g(y_{\mathcal{B}})$ is a solution of $A$. 

Note that while $\tau$ maps instances from $A$ to $B$, $g$ works in the backwards direction, mapping solutions of $B$ to $A$. This will be important for reductions between combinatorial optimization problems.
In some cases, where $g$ is the identity, we call the reduction simply by the instance transformation $\tau$.
Further we denote by $A \leq_p B$ (``$A$ polynomial-time reduces to $B$"), if there exists a polynomial-time reduction from $A$ to $B$.
It is important to note that since the run time of $\tau$ and $g$ are at most polynomial in their inputs, the outputs can be larger than the inputs at most by a factor of ${\rm poly}(|\mathcal{A}|)$, ${\rm poly}(|y_{\mathcal{B}}|)$, respectively.

\begin{figure*}
    \centering
    \includegraphics[width=.6\textwidth]{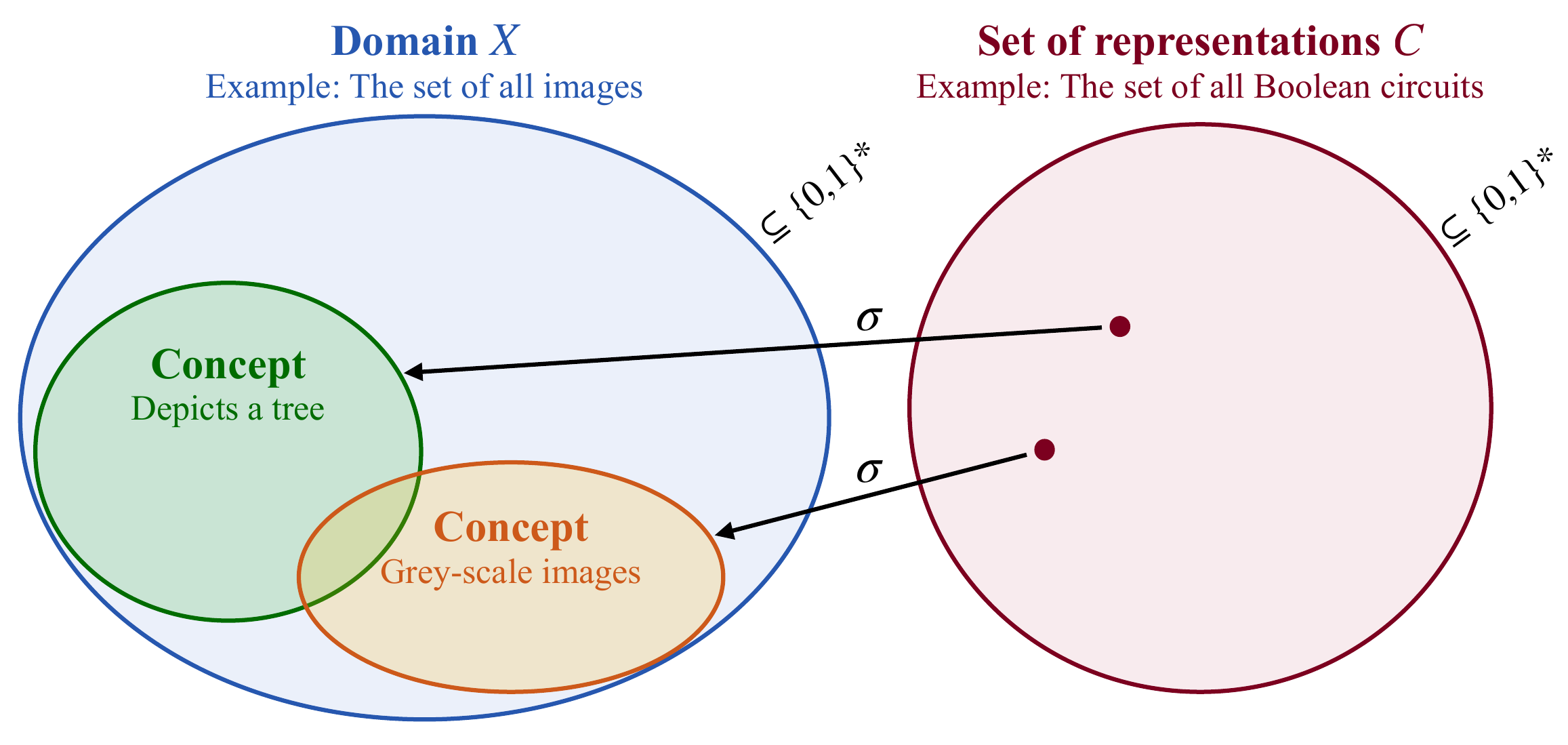}\\
    \caption{
        \ntext{\textbf{The interplay between representations and concept.}}
        The domain $X$ can formally be seen as a set of finite bit strings.
        Concepts are subsets of the domain, which can be described by representations $c \in C$.
        Together with the map $\sigma$, mapping representations to concepts, the tuple $(\sigma, C)$ is called a \emph{representation class}.
    }
    \label{fig:representations}
\end{figure*}

\subsubsection{Reductions among representations}
\label{subsec:reductions}
To understand our proof of the quantum advantage in combinatorial optimization, 
we require polynomial-time reductions among the evaluation problem of representation classes.
Intuitively, these reductions show that one representation class is at least as powerful 
as another and that they can be transformed into each other.
In Ref.~\cite{kearns_boolfunc_1993}, these reductions have been used to derive (classical) 
computationally hard problems for different representations.
First we define a technical construction, the \textit{evaluation problem}.

\begin{definition}[Evaluation problem $\Eval(C)$]
    \label{def:eval}
    {\color{white}.}\\
    \textbf{Instance} The pair $(c, x)$, where $C$ is a representation class over the domain $X$, $c \in C$ is a representation description and $x \in X$.\\
    \textbf{Solution} The result $c(x)$ of $c$ on $x$.
\end{definition}

Let $n \in \mathbb{N}$ and define $\BC_n$ to be the representation class of \emph{polynomially evaluatable Boolean circuits} with domain $X=\{0,1\}^n$ and with depth $O(\log(n))$ and size $O({\rm poly}(n))$, and let $\BC = \cup_{n \geq 1} \BC_n$. 
In a similar manner, define $\BF$ to be the representation class of \emph{Boolean formulae of poly-size}, 
define $\LSTM$ to be the representation class of 
\emph{log-space Turing machines} and finally, define $\DFA$ to be the representation class of 
\emph{deterministic finite automata} \cite{hopcroft_automata_2013} of poly-size. It holds that
\begin{align}
    &\Eval(\BC) \leq_p \Eval(\BF)  ,\label{eq:red1}\\
    &\Eval(\BF) \leq_p \Eval(\LSTM)  ,\label{eq:red2}\\
    &\Eval(\LSTM) \leq_p \Eval(\DFA)  \label{eq:red3}.
\end{align}
Subsequently, we sketch the proof ideas for the three reductions above.
The full proofs can be found in Refs.~\cite{kearns_learningtheory_1994} 
and \cite{kearns_boolfunc_1993}.
From here on after, we consider $n$ to be the size of the input to a Boolean circuit 
in $\BC$.
\begin{itemize}
    \item[(\ref{eq:red1})]
    We denote this polynomial-time reduction by $\tau_1$.
    Recall that $\tau_1$ is the instance transformation algorithm and the solution transformation is the identity.
    Let $c$ be a Boolean circuit in $\BC$ with depth $d=O(\log(n))$ and size $s=O({\rm poly}(n))$.
    Every Boolean circuit can be identified with a directed acyclic graph where each vertex has fan-in at most 2.
    The instance transformation in the reduction goes by starting at the output vertex of $c$ and recursively building the Boolean formula $f$ by walking back through $c$ and substituting clauses in $f$.
    $f$ will then consist of at most $2^d$ clauses over $n$ variables, which is size $O({\rm poly}(n))$.
    Clearly, $c$ and $f$ compute the same function, the reduction (\ref{eq:red1}) holds, since the transformation can be performed by an $O({\rm poly}(n))$-time algorithm.
    We have $\tau_1((c, x)) = (f,x)$.
        
    \item[(\ref{eq:red2})]
    We denote this polynomial-time reduction by $\tau_2$.
    This reduction uses the fact that we can transform any Boolean formula $f$ to a log-space Turing machine $m$ that, on input $x$ computes $m(x) = f(x)$, in time $O({\rm poly}(n))$.
    The details for this transformation can be found in 
    Ref.~\cite{kearns_learningtheory_1994}.
    Again, we denote the operation of this instance transformation algorithm as $\tau_2((f, x)) = (m,x)$.
        
    \item[(\ref{eq:red3})]
    We denote this polynomial-time reduction by $\tau_3$.
    The reduction uses a transformation of 
    log-space Turing machines to \emph{deterministic finite automata} 
    \cite{hopcroft_automata_2013}.
    In particular, for each log-space Turing machine $m$, one can construct a DFA $t$ that on input of 
    polynomially many copies of the original input $x$ simulates $m$ \cite{kearns_learningtheory_1994}.
    Note that in the reduction here, the input is transformed such that $x$ is repeated $p(n)$ 
    many times and then taken as the input to $t$, where $p$ is a polynomial in $n$.
    We thus have $\tau_3((m,x)) = (t, \underbrace{x,\dots ,x}_{{p}(n)\text{ times}})$, such that 
    \begin{equation}
    m(x) = t(x,\dots ,x).
    \end{equation}
\end{itemize}
Sometimes we are only interested in the transformation of the representation description and not in the input $x$. If we say that we transform a representation description $c$ using $\tau_{1,2,3}$, we omit the second input $x$ and simply write $\tau_{1,2,3}(c)$.
Recall that in the reductions above, since the instances are transformed by polynomial-time algorithms, the output instances can be larger than the input at most by a polynomial factor.

\subsubsection{Learning of representations} \label{subsec:learning_representations}
To obtain a classical hardness result for approximation tasks, the work of Ref.~\cite{kearns_boolfunc_1993}
use the so-called \emph{Occam learning framework} \cite{blumer_occam_1987}.
Generally speaking, the Occam learning framework
makes a connection between nearly minimal hypotheses which are consistent with
observations and the ability to generalize from the observed data in the sense of PAC learning.
To introduce this formalism, let $(\sigma, C),(\phi, H)$ be two representation classes over the domain $X \subseteq \bsn$.
In the following we write $C$ for $(\sigma, C)$ and $H$ for $(\phi, H)$ and denote the two representation descriptions $c \in C$ and $h \in H$ as elements of the set of representation descriptions of $(\sigma, C)$ and $(\phi, H)$.
Given a labeled sample 
\begin{equation}
S=\{(x_1, c(x_1)),\dots ,(x_m, c(x_m))\} 
\end{equation}
of $m$ \textit{examples}, we say that $h \in H$ 
is \textit{consistent} with $S$, if and only if $c(x_i) = h(x_i)$ for all $i = 1,\dots ,m$.
The $x_1, \ldots, x_m \in X$ might be drawn at random according to a distribution $D$ over $X$.
Importantly, we denote by $\opt(S)$ the size of the smallest $h \in H$ that is consistent with $S$.
The consistency problem is defined as follows:
\begin{definition}[Consistency problem $\Con(C,H)$ \cite{kearns_boolfunc_1993}]
    \label{def:conprob}
    {\color{white}.}\\
    \textbf{Instance} A labeled sample $S$ of some $c \in C$.\\
    \textbf{Solution} $h \in H$ such 
    that $h$ is consistent with $S$ and 
    $\sizeh$ is minimized.
\end{definition}
We denote by $\Con(C, H)$ the problem of finding a minimal $h \in H$ that is consistent with
some labeled sample $S$ of some $c \in C$ and likewise we call such a minimal consistent $h$ a solution to the consistency problem of an instance $S$ of $\Con(C,H)$.
Occam's razor makes a connection between the consistency problem and the ability to learn one representation class by another.
In this context learning is defined as follows:
Let $0 \leq \epsilon < 1$ and $0 < \delta \leq 1$. 
An $(\epsilon,\delta)$-\textit{probably approximately correct} (PAC) \citep{valiant_pac_1984} learning algorithm for $C$ by $H$ outputs an $h \in H$, such that $\error_{c,D}(h) \leq \epsilon$ with probability at least $1-\delta$ (for all distributions $D$ over $X$
and all $c \in C$).

We are now in the position to introduce the core theorem of this section, which connects the task of PAC learning and an approximation task.
Intuitively, the following theorem states that finding a hypothesis that explains the observed data (i.e., is consistent with $S$) and is \otext{significantly} \ntext{substantially} more compact than the data, is sufficient for PAC learning.

\begin{theorem}[Occam's razor \citep{blumer_occam_1987, kearns_boolfunc_1993}]
    \label{theorem:occam_learner}
    Given a labeled sample S of $c$ of size
    \begin{align}
        m = O\left(\frac{1}{\epsilon} \log \frac{1}{\delta} + \left( \frac{n^{\alpha}}{\epsilon} \log \frac{n^{\alpha}}{\epsilon} \right)^{1/(1-\beta)} \right) \text{,}
    \end{align}
    where the $m$ examples have been sampled independently from $D$ and for some fixed $\alpha \geq 1$ and $0 \leq \beta < 1$, any $h$ that is consistent with $S$ and which satisfies
    \begin{align}
        \sizeh \leq \opt(S)^{\alpha} \sizeS^{\beta} \label{eq:occam}
    \end{align}
    does also satisfy $\error_{c,D}(h) \leq \epsilon$ with probability at least $1-\delta$.\\
\end{theorem}
Here, $\alpha$ and $\beta$
are fixed values for the Occam's razor prescription, the intuition for them
being hinted at in
Ref.\ \cite{blumer_occam_1987}. When $m$ is fixed to a 
sufficiently large number, fulfilling the 
scaling of the above theorem, then $\alpha$ can be 
seen as reflecting the property that $\opt(S)^{\alpha}$ bounds some polynomial in $\opt(S)$ and $\beta$ can hence be viewed as a ``compression parameter''. If $\beta = 0$, we have complete compression. Then the algorithm
provides a consistent hypothesis of 
complexity at most $opt_{Con}(S)^\alpha$,
independent of the sample size.
The sample size needed is then
$m = O(\frac{1}{\epsilon} \log \frac{1}{\delta})$.
For $\beta\rightarrow 1$, we actually have not learned much, since almost all of $S$ can be encoded in $h$.

Then, note that the size of $S$ is a polynomial in $(n,\frac{1}{\epsilon},\frac{1}{\delta})$.
The variable $\alpha$ resembles that $|h|$ must be smaller than some polynomial in the optimal solution size, while $\beta$ forces that $h$ does not simply hard-encode $S$.
Clearly, it follows that any algorithm that for all $c \in C$ and all $D$, on input $S$ sampled according to $D$ of size ${\rm poly}(n,\frac{1}{\epsilon},\frac{1}{\delta})$, outputs an $h \in H$ with $\sizeh$
upper bounded as in 
Theorem \ref{theorem:occam_learner} is a PAC learning algorithm for $C$ by $H$.
Importantly, learning $C$ by $H$ can be interpreted as an approximation task. Specifically, the task is to approximate the optimal solution $\opt(S)$, which is the size of the smallest representation consistent with $S$, by $|h|$, where $h$ is a representation that is also consistent with $S$, for any $S$ of sufficient size.
An algorithm achieving such an approximation within a factor of $\opt(S)^{\alpha-1}|S|^\beta$, for all $S$ with $|S| = {\rm poly}(n,\frac{1}{\epsilon},\frac{1}{\delta})$, is an $(\epsilon,\delta)$-PAC learner for $C$.
In the remainder of this work, when we say that some ``algorithm approximates the solution of the $\Con(C, H)$ problem", we mean that the algorithm outputs an $h$, such that $\sizeh$ approximates $\opt(S)$ by a factor 
$\opt(S)^{\alpha-1}|S|^\beta$,
where $h$ has the important property of being consistent with $S$.
This sense of approximation might seem unnatural, but the $Con$ problem will later be reduced to a combinatorial optimization task, where it is natural to approximate some scalar quantity and satisfy some constraints.

\subsubsection{Formula colouring problem}
We now introduce the \textit{formula colouring problem} (FC) that takes the
centre stage in our later argument. It is a combinatorial optimization problem that 
has originally been introduced in Ref.~\cite{kearns_boolfunc_1993} as  
a generalization of the more common \emph{graph colouring problem}. It is an optimization
problem of the type as is frequently considered in notions of quantum approximate
optimization: In fact, in a subsequent section, we will formulate this 
problem as a problem of minimizing the energy of a commuting local Hamiltonian,
to make that connection explicit. It is one of the main results of this work to 
show a super-polynomial quantum advantage for FC and integer programming.
Let $z_1, \dots, z_m \in \mathbb{N}$ be the \textit{variables} in a Boolean formula, 
each being assigned an integer value, which acts as the integer valued 
\textit{colour} of the variable. That is to say, each of the variables 
$z_1, \dots, z_m$
takes exactly one of the possible values referred to as 
colours. We regard an assignment of colours to the $z_i$ (called a colouring) as a partition 
of the variable set into equivalence classes. That is to say, two variables have the same colour if and only if they are in the same equivalence class. For the FC problem, we consider Boolean formulae $F(z_1, \dots, z_m)$ which consist of conjunctions 
of two types of clauses. On the one hand, these are clauses of the form $(z_i \neq z_j)$.
This is, in fact, precisely of the form as the clauses of the more common 
 \emph{graph colouring problem}. On the other hand, there are clauses of the form
$((z_i = z_j) \rightarrow (z_k = z_l))$. This 
material conditional, as it is called in Boolean logic, 
can equivalently and possibly more commonly be written 
as
\begin{equation}
((z_i \neq z_j) \lor (z_k = z_l)).
\end{equation}
A \textit{colouring} is an assignment of colours to the $z_i$, described by a partitioning $P$ of the variable set into $k$ equivalence classes, i.e., $|P| = k$.
This means that $z_i = z_j$ if and only if they are in the same partition of the $k$ partitions in $P$.
We are now in the position to formulate the formula colouring problem.
\begin{definition}[Formula colouring problem $\FC$ \citep{kearns_boolfunc_1993}]
    {\color{white}.}\\
    \textbf{Instance} A Boolean formula $F(z_1, \dots, z_m)$ which consists of conjunctions of clauses of the form either $(z_i \neq z_j)$ or the form $((z_i = z_j) \rightarrow (z_k = z_l))$.\\
    \textbf{Solution} A minimal colouring $P$ for $F(z_1, \dots, z_m)$ such that $F$ is satisfied.
\end{definition}
A \textit{minimum solution} to the FC problem is a colouring with the fewest colours, i.e., $|P|$ is minimal for all possible 
colourings such that $F$ is satisfied.
The example given in Ref.~\cite{kearns_boolfunc_1993} is the formula
\begin{equation}
(z_1 = z_2) \lor ((z_1 \neq z_2) \land  (z_3 \neq  z_4))
\end{equation}
has as a model the two-colour partition $\{z_1, z_3\}$,
$\{z_2, z_4\}$ and has as a minimum model the one-colour
partition $\{z_1, z_2, z_3, z_4\}$. The formula colouring problem is obviously NP-complete, as the problem is in NP and graph colouring is NP-hard. 

\subsubsection{The RSA encryption function}
\label{subsec:rsa}
Throughout this work, we will make use on the hardness of inverting the RSA encryption function \cite{rivest_RSA_1978}, 
which forms the foundation of the security of the RSA public-key cryptosystem, one of the canonical public-key crypto-systems and presumed to be secure against classical adversaries \cite{goldreich_foundations2_2004}.

Let $N = p \times q$ be the product of two primes $p$ and $q$, both of similar bit-length. Define \emph{Euler's totient function} $\phi$, where $\phi(N)$ is equal to the number of positive integers up to $N$ that are relative prime to $N$. It holds that $x^{\phi(N)} = 1 \mod N$. When two parties, which we refer to as Bob and Alice, wish to communicate via an authenticated but public channel, they can do so as follows: First, Alice generates two primes $p$ and $q$ of similar bit-length and computes their product $N = p \times q$. Then, Alice generates a so-called public-private key pair $(d,e)$, where $d$ is the \textit{secret key} satisfying $d \times e \mod \phi(N) = 1$, and $e$ is the \textit{public exponent}. Alice shares the public key $(e,N)$ with Bob over the public channel. We define the RSA encryption function for a given exponent $e$, a message $x \in \mathbb{Z}_N$, and a \textit{modulus} $N$ as 
\begin{equation}
RSA(x, N, e) = x^e \mod N .
\end{equation}
To encrypt a message $x \in \mathbb{Z}_N$, Bob simply computes the output of the RSA encryption function, given $N$ and $e$. Bob then sends the ciphertext $c = x^e \mod N$ to Alice, who decrypts the ciphertext by computing $c^d \mod N = (x^e)^d \mod N = x^{1+i \times \phi(N)} \mod N = x \mod N$, where the last step follows from the fact that $x^{\phi(N)} = 1 \mod N$ and $e \times d = 1 + i \times \phi(N)$ for some $i \in \mathbb{N}$ because $e \times d \mod \phi(N) = 1$.

The security of the RSA cryptosystem is closely related on the presumed hardness of integer factoring and, more generally, is based on the presumed hardness of \textit{inverting} the RSA encryption function without knowledge of the secret key $d$. That is, there is no known classical polynomial-time algorithm that, given $(\RSA(x,N,e), N, e)$ outputs $x$. On a quantum computer, however, Shor's 
algorithm \cite{shor_factoring_1994} can be used to factor the integer $N$ in polynomial time. This immediately gives rise to a quantum polynomial time algorithm that inverts the RSA encryption function; Simply factor the public modulus using Shor's algorithm, and then compute $\phi(N) = (p-1) \times (q-1)$. Then, one can find a $d$ such that $e \times d \mod \phi(N) = 1$ by using the 
\emph{extended Euclidean algorithm}. In summary, under the standard cryptographic assumption that the RSA encryption function is hard to invert, Shor's algorithm thus gives rise to a computational quantum-classical separation. As we will show, this separation extends to the approximation of combinatorial optimization problems as well. 

Throughout this work, we will make use of the fact that determining the \emph{least significant bit} (LSB) of $x$, given $\RSA(x, N,e)$ is as hard as inverting the RSA encryption function. Formally, 
Alexi et.~al.~\cite{alexi_rsa+rabin_1988} have proven that if there exists a classical polynomial-time algorithm that finds the LSB of $x$, given $RSA(x,N,e)$, then there exists a classical polynomial-time algorithm that inverts the RSA encryption function.

\subsection{Classical hardness of approximation}\label{section:classical_hardness}
To show our quantum advantage, we require a classical hardness result and quantum efficiency result.
In this section, we establish the classical hardness of approximating combinatorial optimization solutions.
We build on the results of Ref.~\cite{kearns_boolfunc_1993}, where the hardness of approximation tasks has been established. Furthermore, their work shows how the these hard-to-approximate problems can be 
reduced to the combinatorial optimization problem of \emph{formula colouring}.
We then extend these results by showing an approximation-preserving reduction from formula colouring to \textit{integer linear programming} (ILP).
These results will constitute the classical hardness part for the quantum-classical separation we show.
\begin{figure*}[t]
    \centering
      \includegraphics[width=.6\textwidth]{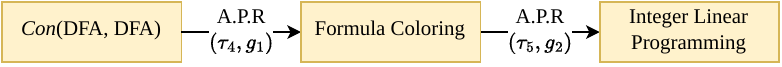}\\
    \caption{\ntext{\textbf{The reduction chain from the consistency problem to combinatorial optimization problems.}} In Section \ref{subsec:classhard_bc}, we introduce Boolean circuits, whose sizes are hard to approximate by $|h|$, where $h$ is a hypothesis that is consistent with a sample labeled by the circuits. This directly implies the approximation hardness of $\Con(\DFA,\DFA)$. In Section \ref{subseq:classhard_fc}, 
    we present an approximation-preserving reduction from $\Con(\DFA, \DFA)$ to formula colouring \cite{kearns_boolfunc_1993}. We then extend the results 
    of Ref.~\cite{kearns_boolfunc_1993} by showing in Section \ref{subsec:classhard_ilp} an approximation-preserving reduction from formula colouring to integer linear programming, yielding the approximation hardness for ILP.}
    \label{fig:class_hard_overview}
\end{figure*}
Fig.~\ref{fig:class_hard_overview} gives a high level overview of the results presented in this section.

\subsubsection{Approximation hardness of the $\Con$ problem}\label{subsec:classhard_bc}\label{subsec:bool_circuits}
In this subsection,     we present the result that approximating the solution of $\Con(\DFA, \DFA)$ is hard using a classical computer~\cite{kearns_boolfunc_1993}. This result is obtained through the assumption that inverting the RSA encryption function is hard, a widely accepted cryptographic assumption.
    To do this, one defines a class of Boolean circuits that essentially decrypt a given RSA ciphertext and output the LSB of the cleartext.
    Intuitively, the authors of Ref.~\cite{kearns_boolfunc_1993} 
    show that, since PAC learning these Boolean circuits is hard (otherwise one would be able to invert RSA), the approximation of these decryption circuits by any polynomially evaluatable representation class in the sense of Theorem \ref{theorem:occam_learner} must also be hard, using a classical computer. They then show that this implies that approximating the solution of $\Con(\DFA,\DFA)$ must also be hard.
    To follow the argumentation in Ref.~\cite{kearns_boolfunc_1993}, let $N \in \mathbb{N}$ and $x \in \mathbb{Z}_N$ and define
     \begin{align}
     \begin{split}
        \powers_N(x) := &x \mod N, x^2 \mod N, x^4 \mod N, \dots\\
            &\dots, x^{2^{\ceil{log(N)}}} \mod N         
     \end{split}
    \end{align}
    as the sequence of the first $\ceil{\log(N)}+1$ square powers of $x$.
    \begin{definition}[Boolean circuit for the LSB of RSA \cite{kearns_boolfunc_1993}]
        \label{definition:rsa_learning_problem}
        Let $\CRSA_n \subset \BC_n$ and $\CRSA = \bigcup_{n \geq 1} \CRSA_{n}$ be the representation class of 
        \emph{log-depth, poly-size Boolean circuits} that, on input $\bcinput$, output $LSB(x)$ for all $x \in \mathbb{Z}_N$.
        Each representation in $\CRSA_{n}$ is defined by a triple $(p,q,e)$ and this representation will be denoted $r_{(p,q,e)}$,
        where $p$ and $q$ are primes of exactly $n/2$ bits and $e \in \mathbb{Z}_N$ and $N = p \cdot q$.
        \\
        An example of $r_{(p,q,e)} \in \CRSA_{n}$ is of the form
        \begin{align}
            &\left(\bcinput, LSB(x)\right),
        \end{align}
        with $x \in \mathbb{Z}_N$.
    \end{definition}
    It is important to note at this point that the calculation of the LSB of $x$, given 
    the input $\bcinput$ can indeed be performed by a $O(\log(n))$-depth, ${\rm poly}(n)$-size Boolean circuit, if the decryption key $d$ is known \citep{kearns_boolfunc_1993}.
    \begin{figure*}[t]
        \centering
          \includegraphics[width=.8\textwidth]{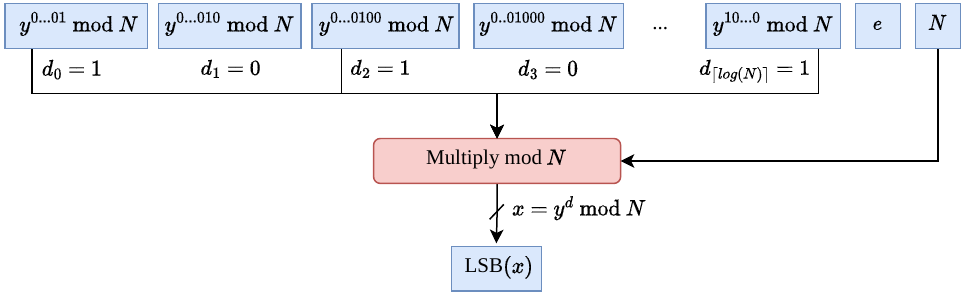}\\
        \caption{
            \ntext{\textbf{A Boolean circuit in the class C-RSA.}}
            The input to the circuit in $\CRSA$ is the power sequence of the RSA ciphertext of $\RSA(x,N,e)=y$.
            The circuit computes the LSB of $x$ by simply performing modular multiplication on
            the $2^i$'th powers of the power sequence where the secret key bit $d_i=1$, for the secret key $d$.
            Thereby the secret key $d$ is hard-wired into the circuit and the decryption $x=y^d \mod N$ is explicitly performed.
            This can be done in an $O(\log(n))$ deep circuit \citep{beame_iteratedproducts_1986}.
        }
        \label{fig:BCRSA}
    \end{figure*}
    In Fig.~\ref{fig:BCRSA}, we depict a schematic picture of such a Boolean circuit in $\CRSA$.
    Since learning the LSB of the cleartext is as hard as inverting the RSA function \cite{alexi_rsa+rabin_1988}, which is widely assumed to be intractable for classical computers, Ref.~\cite{kearns_boolfunc_1993} shows the classical approximation hardness of $\Con(\CRSA, H)$, where $H$ is any polynomially evaluatable representation class.
    The following theorem states that (assuming the classical hardness of inverting RSA) and given some sample $S$ of some $r_{(p,q,e)}\in \CRSA$, no polynomial-time classical algorithm can output a hypothesis $h \in H$ that is consistent with $S$ and only polynomially larger than the smallest possible hypothesis.
    \begin{theorem}[Classical approximation hardness of $\CRSA$ \citep{kearns_boolfunc_1993}]
    \label{theorem:classical_hardness}
        Let $H$ be any polynomially evaluatable representation class. 
        Assuming the hardness of inverting the RSA function, there exists 
        no classical probabilistic polynomial-time algorithm
        that on input an instance $S$ of $\Con(\CRSA, H)$ finds a solution $h \in H$ that is consistent with $S$ and approximates the size $\opt(S)$ of the optimal solution by
        $$|h| \leq (opt_{Con}(S))^{\alpha} |S|^{\beta}$$
        for all $S$ and any $\alpha \geq 1$ and $0 \leq \beta < 1$.
    \end{theorem}
    Since $|S| = n\times m = n \times {\rm poly}(n,\frac{1}{\epsilon}, \frac{1}{\delta})$ and $\alpha \geq 1$ we get that the optimal size $\opt(S)$ cannot be approximated up to a polynomial factor, holding for all classical probabilistic polynomial-time algorithms, where the sense of approximation is explained in detail in Section \ref{subsec:learning_representations}.
    
    \subsubsection{Classical approximation hardness for more representation classes}
    \label{subsubsec:classical-approximation-hardness}
    Furthermore, Ref.~\cite{kearns_boolfunc_1993} shows that the approximation hardness of $\CRSA$ implies approximation hardness for \emph{Boolean formulae}, \emph{log-space Turing machines} and DFAs.
    In particular, let $\BFRSA$ be the class of \emph{Boolean formulae} that we obtain when we reduce every instance in $\CRSA$ using $\tau_{1}$, i.e., $\BFRSA = \{F \mid (F,x) = \tau_{1}(I) \text{ and } I \text{ instance of } \Eval(\CRSA)\}$.
    In a similar manner, $\LSTMRSA$ is the class of \emph{log-space Turing machines} that we obtain when we reduce $\BFRSA$ using $\tau_{2}$ and finally, $\DFARSA$ is the class of DFAs that we obtain when using $\tau_{3}$ on $\BFRSA$.
    Since the evaluation problem of resulting representations are poly-time reducible to each other and are at most polynomially larger, the following holds \citep{kearns_boolfunc_1993}:
    \begin{theorem}[Classical approximation hardness of more representations \citep{kearns_boolfunc_1993}]
    \label{theorem:classical_hardness_more}
        Let $H$ be any polynomially evaluatable representation class.
        Assuming the hardness of inverting the RSA function, there exists no classical probabilistic polynomial-time algorithm
        that, on input an instance $S$ of (a) $\Con(\BFRSA, H)$, (b) $\Con(\LSTMRSA, H)$, or 
        (c) $\Con(\DFARSA, H)$, finds a solution $h \in H$ that is consistent with $S$ and approximates the size $\opt(S)$ of the optimal solution by
        $$|h| \leq (opt_{Con}(S))^{\alpha} |S|^{\beta}$$
        for all $S$ and any $\alpha \geq 1$ and $0 \leq \beta < 1$.
    \end{theorem}
    Specifically, note that approximating the solution of $\Con(\DFARSA,\DFA)$ is at least as hard as to approximate the solution of $\Con(\DFARSA, H)$.
    
\subsubsection{Approximation hardness of formula colouring} \label{subseq:classhard_fc}
In this work we are interested in showing a quantum advantage for approximating the solution of combinatorial optimization problems.
Therefor\ntext{e}, we require a classical approximation hardness result for a combinatorial optimization problem.
To that end, the work of Ref.~\cite{kearns_boolfunc_1993} gives an approximation-preserving reduction from the $\Con(\DFA, \DFA)$ problem to the \emph{formula colouring problem}, which is a combinatorial optimization problem.
We denote the approximation-preserving reduction from $\Con(\DFA, \DFA)$ to $\FC$ by $(\tau_4, g_1)$, where we will explicitly give the construction of the instance transformation $\tau_4$, which maps an instance $S$ of $\Con(\DFA, \DFA)$ to an instance $F_S$ of $\FC$.
First, observe that $S$ contains the examples 
$(w_1, b_1), (w_2, b_2), \dots, (w_m, b_m)$ where $w_i \in \{0,1\}^{k}$ and the labels $b_i \in \{0,1\}$.
The formula $F_S$ will be over variables $z_{i}^{j}$, where $1 \leq i \leq m$ and $0 \leq j < k$.
Essentially, each variable $z_i^j$ will correspond to the state that a consistent DFA would be in after reading the $j$-th bit of $w_i$.

We now give the construction for the formula $F_S$:
For each $i_1, i_2$ and $j_1, j_2$, such that $0 \leq j_1, j_2 < k$ and $w_{i_1}^{j_1+1} = w_{i_2}^{j_2+1}$, we add the predicate
\begin{align}
      ((z_{i_1}^{j_1} = z_{i_2}^{j_2}) \rightarrow (z_{i_1}^{j_1+1} = z_{i_2}^{j_2+1}))
\end{align}
to the conjunctions in $F_S$. 
Intuitively, this encodes that for two inputs $w_{i_1},w_{i_2}$, a DFA that is in the same state for both inputs and then reads the same symbol for both those strings next, the resulting state should also be the same.
To ensure the DFA is consistent with the labels of the sample as well, for each $1 \leq i_1, i_2 \leq m$, such that $b_{i_1} \neq b_{i_2}$, we add the predicate
\begin{align}
    (z_{i_1}^k \neq z_{i_2}^k)
\end{align}
to the conjunctions in $F_S$.
Those clauses encode the fact that for different labels, the states (after reading the whole input) of a consistent DFA must be different, since any state can either only accept or reject.

If $|S|$ is the number of bits in $S$, the resulting $F_S$ consists of $\Theta(|S|^2)$ many clauses.
For the solution transformation $g_1$, as well as the proof that this reduction is indeed correct we refer to the proof in Ref.~\cite{kearns_boolfunc_1993}.
%
It is important to note that by the construction above, the bits of the examples are now encoded in the clauses of $F_S$ together with the correct working of the DFA and the solution (the structure of the minimal DFA) is the minimal colouring of $F_S$.
Due to the results of Ref.~\cite{kearns_boolfunc_1993}, 
the following theorem holds:

\begin{theorem}[Reduction of $Con(\DFA, \DFA)$ to $\FC$ \cite{kearns_boolfunc_1993}]
    \label{theorem:dfa_fs_colouring}
    There is a polynomial time algorithm $\tau_{4}$ that on input an instance $S$ of the problem $Con(\DFA, \DFA)$ outputs an instance $F_S$ of the \emph{formula colouring problem} such that $S$ has a $k$-state consistent hypothesis $M \in \DFA$ if and only if $F_S$ has a colouring $P$, with $|P| = k$.
\end{theorem}
Note that the algorithm $\tau_4$ is precisely the instance transformation of the reduction $(\tau_4, g_1)$ and we have:
\begin{align}
    Con(\DFA, \DFA) \leq_p \FC .\label{eq:red4}
\end{align}
In particular, it holds that
\begin{align}
    Con(\DFARSA, \DFA) \leq_p \FCRSA \label{eq:red5},
\end{align}
where $\FCRSA$ is the class of \emph{formula colouring problems} that result out of running $\tau_4$ \ntext{(introduced in section \ref{subseq:classhard_fc})} on the instances in the problem $\Con(\DFARSA, \DFA)$.
In particular, $g_1$ transforms the minimal solution of $\FC$ into the minimal solution of $\Con(\DFA,\DFA)$, thus $opt_{\FC}(F_S) = opt_{\Con}(S)$ (due to Theorem \ref{theorem:dfa_fs_colouring}) and $|F_S| = \Theta(|S|^2)$.
From those two facts, it follows that finding a valid colouring $P$ of $F_S$, such that $|P| \leq opt_{FC} (F)^{\alpha} |F|^{\beta'}$ would contradict Theorem \ref{theorem:classical_hardness_more}, for the parameter range $\alpha \geq 1$, $0 \leq \beta' < 1/2$.
Thus, the reduction $(\tau_4,g_1)$ preserves the approximation 
hardness of $\Con(\DFARSA, \DFA)$ in the sense of the following theorem \cite{kearns_boolfunc_1993}:
\begin{theorem}[Classical hardness of approximation for \emph{formula colouring} \cite{kearns_boolfunc_1993}]
        \label{theo:classical_hardness_fcrsa}
        Assuming the hardness of inverting the RSA function,
        there exists no classical probabilistic polynomial-time algorithm
        that on input an instance $F_S$ of $\FCRSA$ finds a valid colouring $P$ that approximates the size $opt_{\FC}(F_S)$ of the optimal solution by
        \begin{equation}
            |P| \leq opt_{\FC} (F)^{\alpha} |F|^{\beta}
        \end{equation}
        for any $\alpha \geq 1$ and $0 \leq \beta < 1/2$.
\end{theorem}

In a similar mindset, we present an approximation preserving reduction of $\FC$ to the \textit{integer linear programming problem} in the subsequent section.

\subsubsection{Approximation hardness of integer linear programming}
\label{subsec:classhard_ilp}
In this section, we show an approximation-preserving reduction of the \emph{formula colouring problem} to the problem of \emph{integer linear programming}.
ILP is an NP-complete problem in which many practically relevant 
combinatorial optimization tasks are formulated, such as planning or scheduling tasks 
\cite{wolsey_integer_1999}.
The problem is to minimize (or maximize) an objective function that depends on integer variables.
Additionally, there are constraints on the variables that need to be followed.
Let us define an ILP problem within our formalism:
\begin{definition}[Integer linear programming problem ($\ILP$)]
    \label{def:ilp}
    {\color{white}.}\\
    \textbf{Instance} An \ntext{linear} objective function over integer variables subject to \ntext{linear} constraints of the variables.\\
    \textbf{Solution} A valid assignment of the variables under the constraints, such that the objective function is minimal.
\end{definition}
We now show the reduction $(\tau_5, g_2)$ of \emph{formula colouring} to ILP, by first giving the instance transformation $\tau_5$:

Let $F(z_1, \dots, z_M)$ be a formula colouring instance over variables $z_1, \dots, z_M \in \mathbb{N}$ which is a conjunction of $Q$ clauses of the form $(z_u \neq z_v)$ and $R$ clauses of the form $((z_u = z_v) \rightarrow (z_k = z_l))$ (which is equivalent to $((z_u = z_v) \lor (z_k \neq z_l))$).
For $1 \leq u,i \leq M$ and $1 \leq j \leq R$ we introduce the ILP variables $w_i, x_{u,i}, a_j, b_j, s_j \in \{0,1\}$ and $1 \leq \hat{z}_u \leq M$, where $\hat{z}_u$ resembles the variable $z_u$ in $F$ and $w_i$ indicates if the $i$'th colour is used and $x_{u,i}$ indicates if the variable $\hat{z}_u = i$ and $a_j,b_j,s_j$ are helper variables.

It is important to note that for some $k$-colouring $P=\{P_1,\dots,P_k\}$ of $F$, the clause $(z_u = z_v)$ in $F$ is true iff $z_u,z_v \in P_i$ for some colour $i$.
On the other hand, the clause $(z_u \neq z_v)$ in $F$ is true iff $z_u \in P_i$ and $z_v \centernot\in P_i$ for some colour $i$.
In our ILP construction, we introduce an analogue variable to $z_u$, namely $\hat{z}_u$, where $\hat{z}_u$ directly takes as value the colour $i$, i.e., $\hat{z}_u = i$ iff $z_u \in P_i$.

By our construction, we get the \emph{integer linear programming} problem $\ILP_F$
\begin{equation}
    \text{minimize}  \displaystyle\sum\limits_{1 \leq i \leq M}  w_i 
\end{equation}
subject to the following constraints,
\begin{align}
        \label{eq:ilp1} &\text{for all $u,i \in \{1, \dots, M\}$,} & (x_{u,i} = 1) \Longleftrightarrow (\hat{z}_u = i) ,\\
        \label{eq:ilp2} &\text{for all $u \in \{1, \dots, M\}$,} & \sum_{i=1}^M x_{u,i} = 1 ,\\
        \label{eq:ilp3} &\text{for all $u,i \in \{1, \dots, M\}$,} & x_{u,i} \leq w_i, \\
        \label{eq:ilp4} &\text{for all $Q$ clauses $(z_u \neq z_v)$ and all $i \in \{1, \dots, M\}$,} & x_{u,i} + x_{v,i} \leq 1 ,\\
        \label{eq:ilp5} &\text{for all $R$ clauses $((z_u \neq z_v) \vee (z_k = z_l))$ with $j \in \{1, \dots, R\}$,} & (a_j=1) \Longleftrightarrow (\hat{z}_k = \hat{z}_l) ,\\
        \label{eq:ilp6} & & (b_j =1) \Longleftrightarrow (\hat{z}_u \neq \hat{z}_v) ,\\
        \label{eq:ilp7} & & s_{j} = (a_{j} \vee b_{j}) ,\\
        \label{eq:ilp8} & & s_j \geq 1 ,\\
        \nonumber \\
		\label{eq:ilp9} & \text{and } w_i, x_{u,i}, a_j, b_j, s_j \in \{0,1\} \text{ and } 1 \leq \hat{z}_u,\hat{z}_v,\hat{z}_k,\hat{z}_l \leq M.
\end{align}
Before explaining the constraints, let us note that for the sake of understanding, we display here logical clauses in (\ref{eq:ilp1}), (\ref{eq:ilp5}), (\ref{eq:ilp6}) and (\ref{eq:ilp7}), even though they are technically not ILP constraints.
We refer to 
Section \ref{subsec:app}
on how the logical clauses in (\ref{eq:ilp1}), (\ref{eq:ilp5}), (\ref{eq:ilp6}) and (\ref{eq:ilp7}) are 
concretely converted to inequality constraints.

We define the binary variable $w_i$ to be $1$ iff colour $i$ is used.
Hence, the minimization task at hand over the $w_i's$ corresponds to finding the minimal colouring of $F$.
Constraint (\ref{eq:ilp1}) defines the binary variable $x_{u,i}$ to be $1$ iff $\hat{z}_u = i$, i.e., indicating that $z_u \in P_i$.
Constraint (\ref{eq:ilp2}) ensures that any variable is assigned to exactly one colour.
Constraint (\ref{eq:ilp3}) ensures that if there is some $z_u \in P_i$, then $w_i = 1$, since colour $i$ is used.
Constraint (\ref{eq:ilp4}) encodes the $(z_u \neq z_v)$ clauses in $F$, i.e., that $z_u,z_v$ are not assigned the same colour.
Constraints (\ref{eq:ilp5}), (\ref{eq:ilp6}), (\ref{eq:ilp7}) and (\ref{eq:ilp8}) encode the $((z_u \neq z_v) \vee (z_k = z_l))$ clauses in $F$.

In total, we get $M(4M+Q+1)+12R$ constraints and $2M(M+1)+5R$ variables, which are polynomial in the size of $F$.
Thus $\tau_5$ is indeed computable in polynomial time.
Now the solution transformation $g_2$ simply works by partitioning the variables $z_u, z_v$ into the same set iff $\hat{z}_u = \hat{z}_v$. Clearly, $g_2$ is computable in polynomial time.
We show that $(\tau_5, g)$ is indeed a reduction of $\FC$ to $\ILP$ by proving an even stronger result:

\begin{theorem}[Reduction of FC to ILP]\label{theo:ilp_reduction}
    Let $\tau_5$ be a polynomial-time algorithm that on input an instance $F(z_1,\dots,z_M)$ of the 
    \emph{formula colouring problem} $\FC$ outputs an instance $\ILP_F$ of the \emph{integer linear 
    programming} problem.
    Let $g_2$ be a polynomial-time algorithm that on input an assignment $A$ of $\ILP_F$ outputs a colouring $P$ of $F$.
    There exist $\tau_5, g_2$, such that $P$ is a valid $k$-colouring of $F$ if and only if $A$ is a valid assignment of the variables in $\ILP_F$ such that the objective function of $\ILP_F$ is $k$.
\end{theorem}
\begin{proof}
    Let $\tau_5$ and $g_2$ be the algorithms described in the beginning of this section.\\
    \underline{$\Longrightarrow$:}
    We first prove that if $F$ has a valid colouring $P$ of $k$ colours, then there exists an assignment $A$ of the variables such that 
    \begin{equation}
    \sum\limits_{1 \leq i \leq M}  w_i = k.
    \end{equation}
    Without loss of generality, assume an ordering of the sets in $P = \{P_1,\dots,P_k\}$.
    Since $P$ is a colouring of $F$, the $P_i$'s are pairwise disjoint. We assign the variables in $\ILP_F$ as follows.
    \begin{align}
        &\text{For all $i \in \{1,\dots,M\}$, } &w_i = \indexf(i \leq k),\\
        &\text{for all $u \in \{1,\dots,M\}$,} & \hat{z}_u = \sum_{i=1}^M \indexf(z_u \in P_i) \times i ,\\
        &\text{for all $u,i \in \{1,\dots,M\}$,} & x_{u,i} = \indexf(z_u \in P_i),\\
        &\text{for all $R$ clauses $((z_u \neq z_v) \vee (z_k = z_l))$ with $j \in \{1, \dots, R\}$,} & a_j = \indexf(\hat{z}_k = \hat{z}_l),\\
        & & b_j = \indexf(\hat{z}_u \neq \hat{z}_v),\\
        & & s_j = a_j + b_j
        .
    \end{align}
    Clearly, the objective function of $\ILP_F$ is $k$.
    It remains to be shown that the constraints in $\ILP_F$ are satisfied.
    First, note that from the variable assignments it follows that $(\hat{z}_u = i) \Longleftrightarrow (\indexf(z_u \in P_i) = 1)$.
    We can then see that the constraint (\ref{eq:ilp1}) is satisfied, since
    \begin{equation}
        (x_{u,i} = 1) \Longleftrightarrow (\indexf(z_u \in P_i) = 1) \Longleftrightarrow (\hat{z}_u = i) .
    \end{equation}
    The constraint (\ref{eq:ilp2}) is satisfied, due to the pairwise disjointedness of the sets in $P$ and we get
    \begin{equation}\sum_{i=1}^M x_{u,i} = \sum_{i=1}^M \indexf(z_u \in P_i) = 1.
    \end{equation}
    Next, we turn our attention to constraint (\ref{eq:ilp3}). To see why this constraint is satisfied observe the following. From the fact that $\sum_{i=1}^M x_{u,i} = 1$ it follows that there is exactly one $i'$, for which $x_{u,i'} = 1$. By the definition of $x_{u,i'}$, we have $\indexf(z_u \in P_{i'}) = 1$.
    Since $P =\{P_1,\dots,P_k\}$ and $P_{i'}$ is not empty, it must hold that $i' \leq k$ and hence by construction $w_{i'} = 1$. For all other $i \neq i'$, we have $x_{u,i} = 0$, and thus $x_{u,i} \leq w_i$ and constraint (\ref{eq:ilp3}) is satisfied.
    The constraint (\ref{eq:ilp4}) is satisfied, since we have
    \begin{equation}
     x_{u,i} + x_{vi} = \indexf(z_u \in P_i) + \indexf(z_v \in P_i) \leq 1
     \end{equation}
    because of the assumption that $P$ is a valid colouring and this constraint occurs only for clauses of the form $(z_u \neq z_v)$.
    The constraints (\ref{eq:ilp5}) and (\ref{eq:ilp6}) are satisfied by definition.
    One can easily see that (\ref{eq:ilp7}) and (\ref{eq:ilp8}) are also satisfied, since $P$ is a valid colouring and these constraints only occur for clauses of the form $((z_u \neq z_v) \vee (z_k = z_l))$.\\
    
    \underline{$\Longleftarrow$:}
    Assume that we are given a valid assignment $A$ of the variables in $\ILP_F$, such that $\sum_{1 \leq i \leq M} w_i = k$.
    Then, we can construct a valid colouring $P$ for the corresponding formula colouring instance $F$.
    To this end, run $g_2$ by partitioning the variables $z_u, z_v$ into the same sets iff $\hat{z}_u = \hat{z}_v$.
    Since 
    \begin{equation}\sum_{1 \leq i \leq M} w_i = k,
    \end{equation}
    there exist $i_1,\dots,i_k$ for which $w_{i_1},\dots, w_{i_k} = 1$.
    Since for all $u \in \{1,\dots,M\}$ we have 
    \begin{equation}\sum_{i=1}^M x_{u,i} = 1 
    \end{equation}
    and $x_{u,i} \leq w_i$, there exist $u_1,\dots,u_k$ for which $x_{{u_1},{i_1}}, \dots, x_{{u_k},{i_k}} = 1$.
    Therefore, since the $u_1,\dots,u_k$ are pairwise different and $i_1,\dots,i_k$ are pairwise different and because $(x_{u,i} = 1) \Longleftrightarrow (\hat{z}_u = i)$, there are $\hat{z}_{u_1}=i_1, \dots, \hat{z}_{u_k}=i_k$ that are different from each other.
    Therefore, if we partition variables $z_u, z_v$ into the same partition iff $\hat{z}_u = \hat{z}_v$, we obtain exactly $k$ partitions.
    Now we need to show that this colouring is a valid colouring for $F$.
    The clauses $(z_u \neq z_v)$ are satisfied since constraints (\ref{eq:ilp4}) and (\ref{eq:ilp1}) are satisfied.
    The clauses $((z_u \neq z_v)\lor(z_k = z_l))$ are satisfied since constraints (\ref{eq:ilp5}), (\ref{eq:ilp6}), (\ref{eq:ilp7}), (\ref{eq:ilp8}) are satisfied. This ends the proof of the reduction.
\end{proof}


Thus, we have that
\begin{equation}
    \FC \leq_p \ILP \label{eq:red6}
\end{equation}
and in particular
\begin{equation}
    \FCRSA \leq_p \ILPRSA,\label{eq:red7}
\end{equation}
where $\ILPRSA$ are the instances of $\ILP$ that we get when we apply $\tau_5$ to all instances of $\FCRSA$.
Since, by the same arguments as in Section~\ref{subseq:classhard_fc} and since $g_2$ transforms the minimal solution of $\ILP$ to the minimal solution of $\FC$ and $|\ILP_F| = \Theta(|F|^2)$,
the reduction $(\tau_5,g_2)$ preserves the approximation hardness of $\FCRSA$, in the sense of
the following theorem.

\begin{theorem}[Classical hardness of approximation for \textit{integer linear programming}]
        \label{theo:classical_hardness_ilprsa}
        Assuming the hardness of inverting the RSA function,
        there exists no classical probabilistic polynomial-time algorithm
        that on input an instance $\ILP_F$ of $\ILPRSA$ finds an assignment of the variables in $\ILP_F$ which satisfies all constraints and approximates the size $opt_{\ILP}(\ILP_F)$ of the optimal solution by
        \begin{equation}
            \sum_{1 \leq i \leq M}  w_i \leq opt_{\ILP} (\ILP_F)^{\alpha} |\ILP_F|^{\beta}
        \end{equation}
        for any $\alpha \geq 1$ and $0 \leq \beta < 1/4$.
\end{theorem}

To give a high-level overview of the hardness results established in this section, we present 
in Fig.~\ref{fig:hardness_argumentchain} the chain of implications.

\begin{figure*}[t]
    \centering
    \includegraphics[scale=.7]{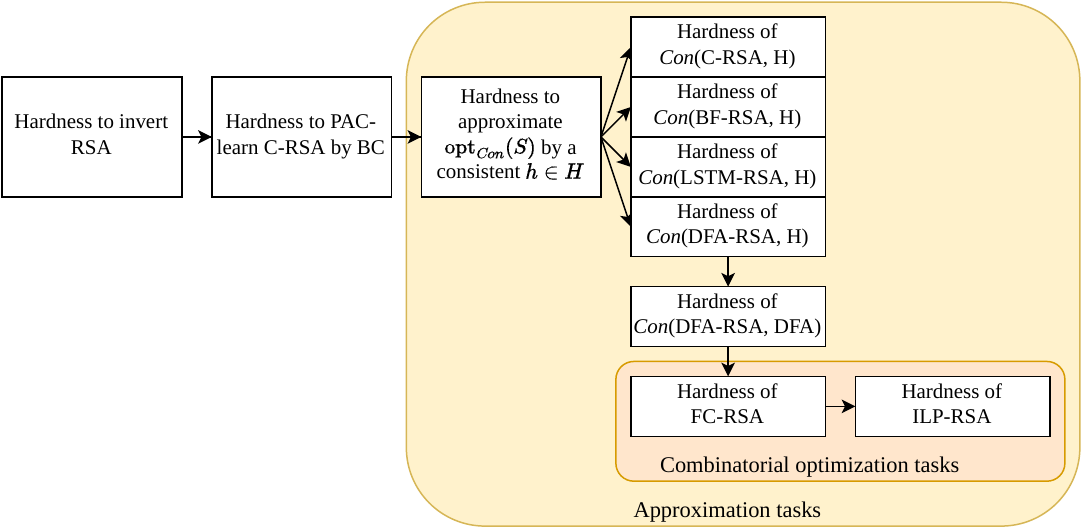}
    \caption{\textbf{The argument chain that propagates the hardness to invert the RSA function to the hardness of approximating combinatorial optimization tasks.}}
    \label{fig:hardness_argumentchain}
\end{figure*}

\subsection{Quantum efficiency} \label{sec:quantum_efficiency}
In the previous section we have presented proofs for the classical hardness of various approximation tasks.
In this section, we turn to showing a quantum advantage by proving that the instances 
resulting from the reductions described in Section \ref{section:classical_hardness} can be solved in polynomial time given access to a fault-tolerant quantum computer.
This yields the desired result of \textit{quantum separation} for natural problems: Under the assumption that inverting the RSA function is hard, quantum computers can find close to optimal solutions to problem instances for which classical computers are incapable of findings solutions of the same quality.

First, we demonstrate that the solutions to instances of $\Con(\CRSA, \BC)$ can be approximated by a polynomial factor in quantum polynomial time leveraging Shor's algorithm.
Later, we show approximation separation results for more ``natural'' problems, namely \emph{formula colouring} and \emph{integer linear programming}.

\begin{theorem}[Quantum efficiency for approximating the solution of $\Con(\CRSA, \BC)$]
    There exists a polynomial-time quantum algorithm that, on input of an instance $S$ of $\Con(\CRSA, \BC)$, finds a consistent hypothesis $h\in \BC$ which approximates the size $\opt(S)$ of the optimal solution by
        \begin{equation}
            |h| \leq \opt(S)^{\alpha}
        \end{equation}
    for all $S$ and for some $\alpha \geq 1$.
\end{theorem}
\begin{proof}
    Let $S$ be an instance of $\Con(\CRSA, \BC)$.
    \begin{algorithm}
        \caption{Approximate the solution of $\Con(\CRSA, \BC)$}
        \label{alg:crsasolver}
        \Input{A labeled sample $S$ of $\CRSA$}
        \Output{The description of a Boolean circuit consistent with $S$}
        \BlankLine
        Pick any example $s \in S$ and read $e, N$ from it\;
        Run \emph{Shor's algorithm} \citep{shor_factoring_1994} to factor $N$ and retrieve $p$ and $q$\;
        Run the extended Euclidean algorithm to compute $d$, such that $d \times e = 1 \mod (p-1)(q-1)$\;
        \tcp{Note that at this point, $d$ is the secret RSA exponent.}
        Output the description of a Boolean circuit that, on input $\bcinput$,  multiplies the $2^i$\ntext{-}th powers together for which the bit $d_i=1$ (thereby hard-wiring $d$ into the circuit), using the iterated products technique \citep{beame_iteratedproducts_1986} and 
        outputs the LSB of the result.
    \end{algorithm}
    In Algorithm \ref{alg:crsasolver}, on input $S$, a hypothesis circuit $h$ is output, which is of size ${\rm poly}(n)$ and which explicitly decrypts a RSA ciphertext given by its power series.
    We know from Section \ref{subsec:bool_circuits} that $h$ is consistent with $S$ and of polynomial size.
    It is clearly the case that $n \leq \opt(S)$ and thus it holds that 
    \begin{equation}
    |h| \leq n^{\alpha} \leq \opt(S)^{\alpha}.
    \end{equation}
\end{proof}
Contrasted with the explicit approximation hardness from Theorem \ref{theorem:classical_hardness}, this yields the super-polynomial advantage of quantum algorithms over classical algorithms for the specific approximation task, namely approximating the optimal consistent hypothesis size by $|h|$ with $h$ consistent with $S$.
We can indeed obtain similar results also for $\Con(\BFRSA, \BF)$, $\Con(\LSTMRSA, \LSTM)$ and $\Con(\DFARSA, \DFA)$.
In particular, given $S$, we can use Algorithm \ref{alg:crsasolver} to obtain a consistent $h$ of $\CRSA$ and then leverage the poly-time instance transformations $\tau_1, \tau_2, \tau_3$, to obtain an at most ${\rm poly}(n)$ larger approximation to the solution of $\Con(\BFRSA, \BF)$, $\Con(\LSTMRSA, \LSTM)$ and $\Con(\DFARSA, \DFA)$. Thus, we obtain the following corollary:

\begin{corollary}[Quantum efficiency for more approximation tasks]
    There exists a polynomial-time quantum algorithm that, on input an instance $S$ of (a) $\Con(\BFRSA, \BF)$, (b)
    $\Con(\LSTMRSA, \LSTM)$, or (c) $\Con(\DFARSA, \DFA)$, finds a consistent hypothesis (a) $h \in \BF$, (b) $h \in \LSTM$, (c) $h \in \DFA$ which approximates the size $\opt(S)$ of the optimal solution by
        $$|h| \leq opt_{Con}(S)^{\alpha}$$
    for all $S$ and for some $\alpha \geq 1$.
\end{corollary}
This again yields super-polynomial advantages of quantum algorithms over 
classical algorithms for approximating the optimal solution size of the consistency problem by the size of a hypothesis that is consistent with the a sample.
While this notion of approximation might seem unnatural, in the subsequent section, we turn our attention to approximating the solution of combinatorial optimization problems, for which it is natural to approximate some optimal scalar value while satisfying certain constraints.

\subsubsection{Quantum advantage for combinatorial optimization}
We now show a super-polynomial quantum advantage for approximating the solution of the combinatorial optimization task of formula colouring.
We have already established the classical approximation hardness of $\FCRSA$ in Theorem \ref{theo:classical_hardness_fcrsa} and give a polynomial-time quantum algorithm for approximating $\FCRSA$ in the proof of the following theorem.

\begin{theorem}[Quantum efficiency for $\FCRSA$]
    \label{theo:quant_eff_fcrsa}
    There exists a polynomial-time quantum algorithm that, on input \ntext{of} an instance $F_S$ of $\FCRSA$, finds a valid
    colouring $P$ such that 
    $$|P| \leq opt_{\FC}(F_S)^{\alpha}$$
    for all $F_S$ and for some $\alpha \geq 1$.
\end{theorem}
\begin{proof}
Let us first describe how any instance $F_S$ of $\FCRSA$ looks like. The overview of the construction of $\FCRSA$ is that we started from class $\CRSA$ of log-depth poly-size Boolean circuits that explicitly decrypt an RSA ciphertext. The representation descriptions in $\CRSA$ were then transformed using $\tau_1$, $\tau_2$ and $\tau_3$ to the class $\DFARSA$.
Thus, recall that any instance $S$ of $\Con(\DFARSA, \DFA)$ is of the form
\begin{align}
    S = \left \{  \left ( \underbrace{ \bigparallel_{l=1}^{p(n)}  \bin\left(\powers_N(\RSA(x_i,N,e)),N,e\right)}_{=w_{i} \text{, of length }p(n)\times(n^2+2n)\text{ bits}}, \underbrace{LSB(x_i)}_{=b_i} \right )  \mid i=1,\dots ,m  \right \} \text{,}
\end{align}
where $\bigparallel$ is the big concatenation of binary strings.
Note that the repetition of $w_i$ $p(n)$ times comes from the reduction $\tau_3$, where for the construction of a DFA that simulates a log-space TM, the input needs to be repeated $p(n)$ times.

Now, $F_S$ is obtained by the reduction $(\tau_4, g_1)$ from Section \ref{subseq:classhard_fc} and $F_S$ is over the variables $z_i^j$, $1 \leq i \leq m$, $1 \leq j \leq p(n)\times(n^2+2n)+1$.
Recall that $z_{i}^{j}$ encodes the state the DFA is in after reading bit $j$ on input $w_{i}$.
By the construction of $F_S$, we know that for each $i_1, i_2$ and $j_1, j_2$, such that 
\begin{equation}
    0 \leq j_1, j_2 < p(n)\times(n^2+2n)+1 
\end{equation}
and $w_{i_1}^{j_1+1} = w_{i_2}^{j_2+1}$, the following predicate 
\begin{align}
    ((z_{i_1}^{j_1} = z_{i_2}^{j_2}) \rightarrow (z_{i_1}^{j_1+1} = z_{i_2}^{j_2+1})) 
\end{align}
occurs in $F_S$. Note that $z_{i}^{0}$ is the starting state of the DFA.

Consider the bit $w_{1}^{n^2+n}$ which is the least significant bit of $N$, for which we know that $LSB(N)=w_{1}^{n^2+n}=1$, since $N$ cannot be even.
We know that for all other bits $w_{i_2}^{j_2+1}$ in $S$ that are equal to $w_{1}^{n^2+n}$, there occurs a predicate
of the form 
\begin{equation}
((z_{1}^{n^2+n-1} = z_{i_2}^{j_2}) \rightarrow (z_{1}^{n^2+n} = z_{i_2}^{j_2+1})) \label{eq:fsform}
\end{equation}
in $F_S$.
Thus, by parsing $F_S$ and looking for all predicates of the form as in (\ref{eq:fsform}), we can infer all bits in $w_i$ given $F_S$, for all $i$. Thus, we can reconstruct all $w_i$'s from $F_S$. Algorithm \ref{alg:S_infer} does exactly this and
runs in time ${\rm poly}(n)$, since there are $O(|S|^2)$ many clauses in $F_S$.
\begin{algorithm}
    \caption{Infer $w_i$ given $F_S$}
    \label{alg:S_infer}
    \Input{An instance $F_S$ of $\FCRSA$, index $i\in\{1,\dots ,m\}$}
    \Output{The bit string $w_i$}
    \BlankLine
    \begin{algorithmic}
        \STATE \tcp{Initialize $w_i$ to the all $0$ string}
        \STATE $w_i \gets 0^{(p(n)\times(n^2+2n))}$\;
        \STATE \tcp{Set the first LSB(N)}
        \STATE $w_i[n^2+n] \gets 1$\;
        \STATE \tcp{Set all other bits of $w_i$ that are also $1$}
        \FOR{$((z_{1}^{n^2+n-2} = z_{i}^{j}) \rightarrow (z_{1}^{n^2+n-1} = z_{i}^{j+1}))$ in $F_S$}
            \STATE $w_i[j+1] \gets 1$\;
        \ENDFOR
        \STATE \KwRet $w_i$\;
    \end{algorithmic}
\end{algorithm}

Remember that our goal in this proof it to give a polynomial-time quantum algorithm that on input an $F_S$ finds a valid colouring of size less than $opt_{\FC}(F_S)^{\alpha}$.
At this point we have described how the instances $F_S$ look like and how we can extract the $w_i$'s from it.
After having obtained a $w_i$ from $F_S$ we read $e$ and $N$ from it and then construct the Boolean circuit $c$, by the same technique employed in Algorithm \ref{alg:crsasolver}.
It is important to note that $c$ is exactly of the form of Boolean circuits in $\CRSA$, from which we originally constructed $\FCRSA$.
When presented the input $\bin\left(\powers_N(\RSA(x_i,N,e)),N,e\right)$, $c$ outputs $LSB(x_i)$. 
We can transform $c$ into a DFA that is consistent with $S$ and then find a colouring for $F_S$ from that DFA. 
Therefor, to obtain a DFA that is consistent with $S$, we run $c$ through the instance transformations $t' =\tau_3 (\tau_2 ( \tau_1(c)))$ to obtain the DFA $t'$ which is consistent with $S$ and of size ${\rm poly}(n)$.
On input $w_i$, $t'$ accepts if $LSB(x_i)=1$ and rejects if $LSB(x_i)=0$.
Now we minimize $t'$ using the standard DFA minimization algorithm \citep{hopcroft_automata_2013} to obtain
the smallest and unique DFA $t$ which accepts the same language as $t'$ and thus is also consistent with $S$ and of minimal size.
This DFA minimization is in principle not needed for the proof, but it is a further optimization step.

We then run Algorithm \ref{alg:colouring} to obtain a colouring for $F_S$ from $t$.
The DFA $t$ consists of the set of states $Q$, the set of input symbols $\Sigma = \{0,1\}$, the set of accepting states $\omega \subseteq Q$, the start state $q_0 \in Q$, and the transition function $\lambda$ that takes as arguments a state and an input symbol and returns a state \cite{hopcroft_automata_2013}. Furthermore, without loss of generality, we fix an ordering of the states in $Q={0,\dots ,k-1}$ with $q_0 = 0$.
\begin{algorithm}
    \caption{Obtain colouring for $F_S$ from $t$}
    \label{alg:colouring}
    \Input{The inputs $w_i$ of $S$ and a DFA $t = (Q, \Sigma, \omega, q_0, \lambda)$ 
    that is consistent with $S$}
    \Output{A valid colouring for $F_S$}
    \BlankLine
    \begin{algorithmic}
        \STATE \tcp{Initialize data structure}
        \STATE $k \gets |Q|$\;
        \STATE $T \gets \texttt{Map\{Int, Set\}}$\;
        \STATE $T[0,\dots ,k-1] \gets \{\}$\;
        \STATE \tcp{Read $w_i$ bit by bit, walk through $t$ and add fill the colours}
        \FOR{$i=1$ \KwTo $m$}
            \STATE \tcp{Begin by the starting state}
            \STATE $c \gets 0$\;
            \STATE $T[c] \gets T[c] \cup \{ z_i^0 \}$\;
            \FOR{$j=1$ \KwTo $p(n) \times (n^2+2n)$}
                \STATE $c \gets \lambda(c, w_i^j)$\;
                \STATE $T[c] \gets T[c] \cup \{ z_i^j \}$\;
            \ENDFOR
        \ENDFOR
        \STATE \KwRet $\{T[0], \dots,  T[k-1]\}$\;
    \end{algorithmic}
\end{algorithm}
We can convince ourselves that the result of Algorithm \ref{alg:colouring} is indeed a valid colouring for $F_S$, since it assigns $z_{i_1}^{j_1}$ and $z_{i_2}^{j_2}$ the same colour if and only if $t$ is in the same state after reading $w_{i_1}^{j_1}$ on input $w_{i_1}$ and after reading $w_{i_2}^{j_1}$ on input $w_{i_2}$.
Therefore, a conjunct 
\begin{equation}
((z_{i_1}^{j_1} = z_{i_2}^{j_2}) \rightarrow (z_{i_1}^{j_1+1} = z_{i_2}^{j_2+1})) 
\end{equation}
cannot be violated since it appears in $F_S$ only if $w_{i_1}^{j_1+1} = w_{i_2}^{j_2+1}$ and by Algorithm \ref{alg:colouring}, if $z_{i_1}^{j_1}$ is assigned the same colour as $z_{i_2}^{j_2}$, then $z_{i_1}^{j_1+1}$ and $z_{i_2}^{j_2+1}$ have the same colour \cite{kearns_boolfunc_1993}.
Additionally, a conjunct 
\begin{equation}
(z_{i_1}^{p(n)\times(n^2+2n)} \neq z_{i_2}^{p(n)\times(n^2+2n)}) 
\end{equation}
cannot be violated since it appears only if $b_{i_1} \neq b_{i_2}$ and if $z_{i_1}^{p(n)\times(n^2+2n)}$ would be assigned the same colour as $z_{i_2}^{p(n)\times(n^2+2n)}$, then $t$ would be in the same state after reading all bits of $w_{i_1}$ and $w_{i_2}$, which is either an accepting or rejecting state, which in turn contradicts that $t$ is consistent with $S$ and $b_{i_1} \neq b_{i_2}$ \cite{kearns_boolfunc_1993}.
%
It follows that the colouring obtained through Algorithm \ref{alg:colouring} is upper bounded by $opt_{\FC}(F_S)^{\alpha}$ for some $\alpha$, since $t$ has polynomial size with the number of states given by $k = |Q| \leq n^{\alpha} \leq \opt(S)^{\alpha} = opt_{\FC}(F_S)^{\alpha}$ and Theorem \ref{theorem:dfa_fs_colouring}.
\end{proof}
Thus, due to Theorem \ref{theo:classical_hardness_fcrsa} and \ref{theo:quant_eff_fcrsa} we have the super-polynomial quantum advantage for approximating a combinatorial optimization solution. 

\ntext{It is interesting to note, that whether an instance $I$ of $\FC$ belongs to the set $\FCRSA$ can be decided in quantum polynomial-time.}
\ntext{To see why, for a given $\FC$ instance $I$, it can be decided in quantum polynomial-time whether the instance is also contained in $\FCRSA$, consider the following algorithm $\mathcal{A}$. First, $\mathcal{A}$ tries to reconstruct the RSA parameters $N$,$e$, and the ciphertext-label pairs from $I$. If these parameters cannot be reconstructed from $I$ (because it does not follow the correct structure), clearly $I \notin \FCRSA$. If $\mathcal{A}$ can reconstruct the respective parameters, then $\mathcal{A}$ constructs a $\Con(\DFA, \DFA)$ instance and then applies the described reduction chain to create an instance of $\FCRSA$. If the resulting instance matches instance $I$, clearly $I \in \FCRSA$ and can therefore be solved by algorithm \ref{alg:colouring}.}

We reuse the techniques employed above to prove the super-polynomial quantum advantage for approximating the optimal solution of an integer linear programming problem, namely $\ILP_{F_S} \in \ILPRSA$.

\begin{theorem}[Quantum efficiency for $\ILPRSA$]
    \label{theo:quant_eff_ilprsa}
    There exists a polynomial-time quantum algorithm that, on input an instance $\ILP_{F_S}$ of $\ILPRSA$, finds a variable assignment
    $A$ that satisfies all constraints and for which the objective function is bounded as
    $$\sum_{1 \leq i \leq M}w_i \leq opt_{\ILP}(\ILP_{F_S})^{\alpha}$$
    for all $\ILP_{F_S}$ and for some $\alpha \geq 1$.
\end{theorem}
\begin{proof}
    Given an instance $\ILP_{F_S}$, one can easily reconstruct $F_S$ from the constraints (\ref{eq:ilp4}) - (\ref{eq:ilp8}) in polynomial time. It is then possible to obtain a valid colouring $P$ of $F_S$ given the routine described in the proof for Theorem \ref{theo:quant_eff_fcrsa}, such that $|P| \leq opt_{\FC}(F_S)$. 
    With $P$, we can get a valid assignment of the variables in $\ILP_{F_S}$ using the routine described in the $\Longrightarrow$-direction in the proof of Theorem \ref{theo:ilp_reduction}. Also due to Theorem \ref{theo:ilp_reduction}, we know that this variable assignment admits the objective function of $\ILP_{F_S}$ to be less than $opt_{\ILP}(\ILP_{F_S})^{\alpha} = opt_{\FC}(F_S)^{\alpha}$.
\end{proof}
Thus, due to the classical approximation hardness from Theorem \ref{theo:classical_hardness_ilprsa}, we encounter a super-polynomial quantum 
advantage for approximating the solution of an integer linear programming problem. It is 
important to stress that the reduction is explicit: That is to say, we can construct
the instances for which one can achieve a quantum advantage of this kind.

\subsection{The optimization problem in terms of a quantum Hamiltonian} \label{subsec:hamiltonian}
The quantum algorithm presented is distinctly not of a variational type, as they are commonly proposed for approximating combinatorial optimization tasks using a quantum computer \cite{farhi_qaoa_2014}.
That said, it is 
still meaningful to formulate the problem at hand as an energy minimization 
problem, to closely connect the findings established here 
to the performance of variational quantum algorithms \cite{cerezo_variational_2021,mcclean_variational_2016}
in near-term quantum computing, as this is the context in which such problems are 
typically stated.
\ntext{It remains to be investigated to which extent the resulting instances can be practically studied and solved on near-term quantum computers.
The aim here is to provide a formal connection from formula coloring and integer linear programming problems to variational quantum algorithms, where}
the problems are commonly stated as 
\emph{unconstrained
binary optimization problems} of the form 
\begin{eqnarray}
	x^*&:=& 
	\text{argmin}_{x\in \{0,1\}^n}
	f(x),
\end{eqnarray}
where $f:\{0,1\}^n \rightarrow \mathbbm{R}$ is an appropriate 
cost function and $x^*$ is a solution bit string of $f$. Particularly common
are \emph{quadratic unconstrained
binary optimization problems},
\begin{eqnarray}
	\text{minimize } f(x) &=& x^TQx ,\\
 x^*&:=& 
	\text{argmin}_{x\in \{0,1\}^n}
	f(x),
\end{eqnarray}
where $Q=Q^T$ is a real symmetric matrix. In fact, it is a well-known result that 
all higher order polynomial
binary optimization problems can be cast into the form of such 
a quadratic unconstrained
binary optimization problem, possibly by adding further auxiliary 
variables;
but it can also be helpful to keep the higher order polynomials.
All such problems can be directly mapped to Hamiltonian problems.
Notably, for quadratic
unconstrained binary optimization problems, the
minimum is equivalent with the ground state energy of the 
\emph{quantum Ising
Hamiltonian} defined on $n$ qubits as
\begin{equation}
	H =
	\sum_{i,j=1}^n
	Q_{i,j} 
	(\mathbbm{1}- Z_i)
	(\mathbbm{1}- Z_j),
\end{equation}
where $Z_j$ is the Pauli-$Z$ operator supported on site labeled $j$. For higher order polynomial problems,
one can proceed accordingly.

Let us, pars pro toto, 
show how the formula colouring problem in the centre of this work 
can be cast into a quartic binary optimization problem.
Let $k\in \mathbbm{N}$ be an upper bound to the
number of colours used for a formula over $m$ variables, with $z_1, \dots , z_m 
\in \{1,\dots, k\}$ being the variables in the formula. We can then make use
of $n=mk$ bits (which then turn into 
$n=mk$ qubits). These bits referred to as $b_{v,c}$ feature the double labels $(v,c)$,
where $v \in \{ 1,\dots, m\}$ labels the vertices and
$ c\in \{1,\dots, k\}$ the colours. If the vertex $v$ is assigned the 
colour $c$, we set $b_{v,c}=1$, and $b_{v,d}=0$ for all $d\neq c$.
To make sure that the solution will satisfy such an
encoding requirement, one adds a penalty of the form
$(1- \sum_{c=1}^k b_{v,c})^2$.
The clauses of the form
$(z_i \neq z_j)$ are actually
precisely like in the graph
colouring problem \cite{tabi_graphcoloring_2020}. This can be incorporated
by penalty terms of the type
$\sum_{c=1}^k b_{z_i, c} b_{z_j,c}$:
Then equal colours are penalized by
energetic terms. The second type of clause requires more thought: 
Exactly if $(z_i = z_j)$ is true
and $(z_j= z_k)$ is false, there should be a Hamiltonian 
penalty. As such, this is a quadratic
Boolean constraint of the form
\begin{equation}
    \Bigl(\sum_{c=1}^k b_{z_i, c} b_{z_j,c}\Bigr)
    \Bigl(1- \sum_{d=1}^k
    b_{z_j, d} b_{z_k,d}\Bigr). 
\end{equation}
Again, this can be straightforwardly be incorporated into a 
commuting classical Hamiltonian involving only terms of the type 
$(\mathbbm{1}- Z_j)$ for suitable site labels $j$, precisely as commonly considered in 
quantum approximate optimization \cite{farhi_qaoa_2014}.
Lastly, to ensure we find a minimal colouring, we can either run 
the quantum optimization algorithm for increasing $k$ and check whether a valid
colouring has been found or one adds additional $k$ qubits $w_c$, $c\in\{1,\dots,k\}$,
which we enforce to be $1$ if colour $c$ is used and $0$ if colour $c$ is not used by
adding the energetic penalty $b_{v,c} - b_{v,c}w_c$ for all $v\in\{1,\dots,m\},c\in\{1,\dots,k\}$.
This corresponds to enforcing the inequality $b_{v,c} \leq w_c$.
We can then add the energetic penalty $\sum_{c=1}^k w_c$ to enforce the optimization algorithm
to find the minimal colouring.
For these reasons, the approximation results proven here motivate the application of quantum optimization techniques for commuting Hamiltonian optimization problems.
Note that this construction is very similar to the integer linear program we proposed in Section \ref{subsec:classhard_ilp} to reduce the formula colouring problem to ILP.

\ntext{
Since any combinatorial optimization problem of the type discussed here can be mapped to a local Hamiltonian it is apparent that the local Hamiltonian problem is NP-hard.
In fact it is even known to be QMA-complete \cite{Kempe_2006}, which is at least as hard as NP.
However, for the $\FCRSA$ instances---which give rise to a specific subclass of local Hamiltonians---it remains to be studied how well the corresponding Hamiltonians can be solved using quantum optimization algorithms in practice.
}
\vspace{\columnsep}
}

{\small
\onecolumngrid
\vspace{\columnsep}

\subsection{Modelling logical clauses as inequality constraints}\label{subsec:app}

In this section, we present some details of proofs that are made reference to 
in the main text.
To model the logical 
Boolean operator
$\lor$, such that $s := (a \lor b)$ for binary variables $s,a,b$, we require the inequality constraints
\begin{align}
    &s \geq a ,\label{eq:ilp_or1}\\
    &s \geq b ,\label{eq:ilp_or2}\\
    &s \leq a + b, \label{eq:ilp_or3}
\end{align}
which is easily seen as being
equivalent.

We are here interested in modelling logical equivalences of the form $(a = 1) \Longleftrightarrow (\hat{z}_u = \hat{z}_v)$ and $(b = 1) \Longleftrightarrow (\hat{z}_u \neq \hat{z}_v)$ for the binary variables $a,b$ and integers $\hat{z}_u,\hat{z}_v$.
For the former, we model the forward and backward implications as follows.

\underline{$(a = 1) \Longrightarrow (\hat{z}_u = \hat{z}_v)$:}
Choose a large enough constant $L$ such that $\hat{z}_u + \hat{z}_v \leq L$, then, since $\hat{z}_u,\hat{z}_v \geq 0$, the following constraints encode the implication.
\begin{align}
    \hat{z}_u &\leq \hat{z}_v + (1-a)L ,\label{eq:ilp10}\\
    \hat{z}_u &\geq \hat{z}_v - (1-a)L .\label{eq:ilp11}
\end{align}
Clearly, the constraints (\ref{eq:ilp10}), (\ref{eq:ilp11}) are satisfied for $\hat{z}_u = \hat{z}_v$ if $a=1$ and for any $\hat{z}_u, \hat{z}_v$ if $a=0$.

\underline{$(\hat{z}_u = \hat{z}_v) \Longrightarrow (a = 1)$:}
Note that this implication is equivalent to $(a \neq 1) \Longrightarrow (\hat{z}_u \neq \hat{z}_v)$, which again is equivalent to $(a \neq 1) \Longrightarrow ((\hat{z}_u > \hat{z}_v) \lor (\hat{z}_u < \hat{z}_v))$, which we will model below.
We introduce a new binary variable $q$, for which, if $a=0$ and $q=1$ then $\hat{z}_u < \hat{z}_v$ and if $a=0$ and $q=0$ then $\hat{z}_u > \hat{z}_v$. This can be modelled by the constraints
\begin{align}
    \hat{z}_u &< \hat{z}_v + (1-q+a)L ,\label{eq:ilp12}\\
    \hat{z}_u &> \hat{z}_v - (q + a)L.\label{eq:ilp13}
\end{align}
The constraints (\ref{eq:ilp12}), (\ref{eq:ilp13}) are satisfied for $\hat{z}_u \neq \hat{z}_v$ if $a=0$ and for any $\hat{z}_u, \hat{z}_v$ if $a=1$.
The variable $q$ essentially indicates if $\hat{z}_u < \hat{z}_v$ or if $\hat{z}_u > \hat{z}_v$ when $a=0$ and can be ignored after the optimization process.
In a similar manner to the constraints above, we can model $(b = 1) \Longleftrightarrow (\hat{z}_u \neq \hat{z}_v)$ as\\
\begin{align}
    \hat{z}_u &< \hat{z}_v + (2-q'-b)L, \label{eq:ilp14}\\
    \hat{z}_u &> \hat{z}_v - (1+q'-b)L ,\label{eq:ilp15}\\
    \hat{z}_u &\leq \hat{z}_v + bL ,\label{eq:ilp16}\\
    \hat{z}_u &\geq \hat{z}_v - bL,\label{eq:ilp17}
\end{align}
in terms of inequality constraints.
\vspace{\columnsep}
}


\section*{Acknowledgements}
\subsection*{Funding}
This work has been supported by the Einstein Foundation (Einstein 
Research Unit on Quantum Devices), for which this is a joint node project, and the MATH+ Cluster of Excellence. It has also received funding from the BMBF (Hybrid), the BMWK (EniQmA), the Munich Quantum Valley (K-8), the DFG (CRC 183). The authors acknowledge the financial support by the Federal Ministry of Education and Research of Germany in the programme of ``Souverän. Digital. Vernetzt.'' Joint project 6G-RIC, project identification number: 16KISK030. Finally, it has also received funding from the QuantERA (HQCC) and the ERC (DebuQC).

\subsection*{Author Contributions}

N.P.\ and V.U.\ developed the theory and carried out the calculations, with relevant contributions from all authors. They also wrote the manuscript with support from J.E.
F.W.\  verified the methods and results.
J.E.\  supervised the project. J.E.\  derived the variational Hamiltonian.
J.P.S.\  devised the main conceptual ideas.

\subsection*{Competing Interests}
All authors declare that they have no competing interests.

\subsection*{Data and Materials Availability}
All data needed to evaluate the conclusions in the paper are present in the paper and/or the Supplementary Materials.

\end{document}